\documentclass[a4paper]{article}
\usepackage{color}
\usepackage{amsmath}
\usepackage{amsthm}
\usepackage{amssymb}
\usepackage{amsmath}
\usepackage{mathcomp}
\usepackage{dsfont}

\newtheorem{definition}{Definition}[section]
\newtheorem{lemma}[definition]{Lemma}
\newtheorem{proposition}[definition]{Proposition}
\newtheorem{theorem}[definition]{Theorem}
\newtheorem{remark}[definition]{Remark}
\newtheorem{corollary}[definition]{Corollary}

\numberwithin{equation}{section}

\def\eps{\varepsilon}

\def\tr{\mathrm{tr}}

\title{On the correlation energy of interacting fermionic systems in the mean-field regime}
\author{Christian Hainzl, Marcello Porta, Felix Rexze \\
\\
Mathematics Department, Eberhard Karls Universit\"at T\"ubingen, \\ Auf der Morgenstelle 10, 72076 T\"ubingen}

\begin{document}

\maketitle

\begin{abstract}
We consider a system of $N\gg 1$ interacting fermionic particles in three dimensions, confined in a periodic box of volume $1$, in the mean-field scaling. We assume that the interaction potential is bounded and small enough. We prove upper and lower bounds for the correlation energy, which are optimal in their $N$-dependence. Moreover, we compute the correlation energy at leading order in the interaction potential, recovering the prediction of second order perturbation theory. The proof is based on the combination of methods recently introduced for the study of fermionic many-body quantum dynamics together with a rigorous version of second-order perturbation theory, developed in the context of non-relativistic QED.
\end{abstract}

\section{Introduction}

Physical systems of relevance in nature are usually formed by a very large number of particles, which makes them very difficult to study starting from the fundamental equations of many-body quantum mechanics. For this reason, physicists introduced effective theories, that are expected to capture the relevant features of complex quantum systems, in suitable scaling limits. Here we shall focus on the ground-state properties of interacting fermionic systems in the mean-field regime. In particular, the present paper is devoted to the computation of the {\it correlation energy}, which will be defined below. 

Before doing this, as a motivating example, let us consider a model for a metal where the nuclei are, for simplicity, represented by a uniform charge background and the electrons interact via the Coulomb potential. Such model is called {\it jellium}. Mathematically, we consider a system of $N$ fermions confined in $\Lambda \subset \mathbb{R}^{3}$, with the following Hamiltonian:
\begin{equation}
H_{N,\Lambda}^{\text{J}} = \sum_{j=1}^{N} (-\Delta_{j} - V_{\text{back}}(x_{j})) + \sum_{1\leq i<j\leq N} \frac{1}{|x_{i} - x_{j}|}\;,
\end{equation}
with $V_{\text{back}}(x) = \frac{N}{|\Lambda|} \int_{\Lambda} dy\, |x - y|^{-1}$ playing the role of background potential. The Hamiltonian $H_{N,\Lambda}^{\text{J}}$ has to be understood as acting on the antisymmetric subspace of the square integrable functions,
\begin{equation}
L^{2}_{\text{a}}(\Lambda^{N}) = \{ \psi_{N} \in L^{2}(\Lambda^{N}) \mid \psi(x_{1}, \ldots, x_{N}) = \text{sgn}(\pi) \psi(x_{\pi(1)}, \ldots, x_{\pi(N)}) \}\;,
\end{equation}
with $\pi$ a permutation of $\{1, \ldots, N\}$. Let us define the ground state energy density of the system as:
\begin{equation}\label{eq:erhoJ}
e^{\text{J}}(\rho) = \lim_{\substack{|\Lambda|\to \infty \\ N/|\Lambda| \to \rho}} \frac{1}{|\Lambda|}\Big( \inf_{\psi_{N} \in L^{2}_{\text{a}}(\Lambda^{N})} \frac{\langle \psi_{N}, H_{N,\Lambda}^{\text{J}} \psi_{N}\rangle}{\langle \psi_{N}, \psi_{N} \rangle} + \frac{\rho^{2}}{2} \int_{\Lambda \times \Lambda} dxdy\, \frac{1}{|x-y|}\Big)\;.
\end{equation}
The last term in Eq. (\ref{eq:erhoJ}) takes into account the Coulomb energy of the static background.  We are interested in the asymptotic behavior of $e^{\text{J}}(\rho)$ in the high density regime, $\rho \gg 1$. A simple upper bound for $e^{\text{J}}(\rho)$ is obtained by taking the noninteracting ground state, the filled Fermi sea, as a trial state. One gets, see {\it e.g.} \cite{GS}:
\begin{equation}
e^{\text{J}}(\rho) \leq \frac{3}{5}c_{\text{TF}} \rho^{5/3} - c_{\text{D}} \rho^{4/3}\;,
\end{equation}
for $c_{\text{TF}} = (6\pi ^{2})^{2/3}$ and $c_{\text{D}} = (2\pi)^{-3} c_{\text{TF}}^{2}$. We introduce the correlation energy $\mathcal{C}(\rho)$ as the difference of the full many-body ground state energy minus the energy the noninteracting Fermi sea:
\begin{equation}
e^{\text{J}}(\rho) = \frac{3}{5}c_{\text{TF}} \rho^{5/3} - c_{\text{D}} \rho^{4/3} + \mathcal{C}(\rho)\;.
\end{equation}
In the fifties, in a series of important papers, see {\it e.g.} \cite{BP}, Bohm and Pines predicted that, for $\rho \gg 1$  (see also \cite{Macke} for earlier results):
\begin{equation}\label{eq:BP}
\mathcal{C}(\rho) = -c_{\text{corr}} \rho \log \rho + O(\rho)\;, 
\end{equation}
for some explicit constant $c_{\text{corr}} > 0$. The method of Bohm and Pines is based on the inclusion of correlations on the free Fermi sea via suitable canonical transformations, that allow to take into account energy excitations around the Fermi level in a nonperturbative way. Similar methods have then been used to compute the $O(\rho)$ correction to $\mathcal{C}(\rho)$ by Gell-Mann and Brueckner \cite{GB}, and Sawada \cite{Sa}. See \cite{FW, GV} for pedagogical reviews.

Proving the rigorous validity of the prediction (\ref{eq:BP}) is a long-standing open problem in mathematical many-body theory. So far, the only available rigorous result on $\mathcal{C}(\rho)$ is the bound obtained by Graf and Solovej in \cite{GS}. They proved that for all $0\leq \delta < 1/15$ there exists a constant $K_{\delta} > 0$ such that:
\begin{equation}\label{eq:GS0}
0 \geq \mathcal{C}(\rho) \geq -K_{\delta}\rho^{4/3 - \delta}\;.
\end{equation}
The upper bound trivially follows from the variational principle. The proof of the lower bound is based on correlation inequalities for the many-body interaction, initially introduced in the work of Bach \cite{Ba, Ba2} for atoms and molecules, combined with semiclassical ideas. Unfortunately, the bounds (\ref{eq:GS0}) are far from the expected result (\ref{eq:BP}). Improving on (\ref{eq:GS0}) is a difficult problem, which requires a better understanding of the correlations among the fermionic particles.

In this paper, we shall consider the problem of computing the correlation energy in a simplified regime, the mean-field scaling. We consider a system of $N$ interacting fermionic particles in three dimensions, confined in a box $\Lambda = [0;1]^{3}$ with periodic boundary conditions. The Hamiltonian of the system is, for $\eps = O(N^{-\frac{1}{3}})$:
\begin{equation}\label{eq:HNdef}
H_{N} = \sum_{j=1}^{N} -\eps^{2}\Delta_{j} + \frac{1}{N}\sum_{1\leq i<j\leq N} v(x_{i} - x_{j})\;,
\end{equation}
acting on $L^{2}_{\text{a}}(\Lambda^{N})$. The interaction potential $v$ is supposed to be bounded and to vary on a scale $O(1)$; therefore, thanks to the choice of the coupling constant, the typical interaction energy is $O(N)$. Notice the presence of the semiclassical parameter $\eps^{2}$ in front of the Laplacian. It is well-known that in this high density regime the unscaled kinetic energy of the system grows as $N^{5/3}$, by the antisymmetry of the wave function. A lower bound on the kinetic energy compatible with this behavior is provided by the Lieb-Thirring kinetic energy inequality, see Chapter 4 of \cite{LS} for a review, stating that, for a suitable constant $C_{\text{LT}}>0$ and for any antisymmetric wave function $\psi_{N} \in L_{\text{a}}^{2}(\mathbb{R}^{3N})$:
\begin{equation}
\langle \psi_{N}, \sum_{j=1}^{N} -\Delta_{j} \psi_{N} \rangle \geq C_{\text{LT}} \int dx\, \rho_{\psi}(x)^{\frac{5}{3}}\;,
\end{equation}
with $\rho_{\psi}(x) = N \int dx_{2}\ldots d x_{N}\, |\psi_{N}(x, x_{2}, \ldots, x_{N})|^{2}$ the density associated to the wave function $\psi_{N}$. Thus, in the present scaling the kinetic energy and the interaction energy corresponding to (\ref{eq:HNdef}) are of the same order. This choice of parameters defines the fermionic mean-field scaling.

We shall focus on the ground state energy of the system:
\begin{equation}
E_{N} = \inf_{\psi_{N}\in L^{2}_{\text{a}}(\Lambda^{N})} \frac{\langle \psi_{N}, H_{N} \psi_{N}\rangle}{\langle \psi_{N}, \psi_{N} \rangle}\;.
\end{equation}
We are interested in comparing this quantity with the energy of the noninteracting ground state, given by:
\begin{equation}\label{eq:FFS}
\psi_{\text{Slater}}(x_{1}, \ldots, x_{N}) = \frac{1}{\sqrt{N!}}\sum_{\pi} \text{sgn}(\pi) f_{k_{1}}(x_{\pi(1)})\cdots f_{k_{N}}(x_{\pi(N)})\;,
\end{equation}
where $f_{k_{i}}(x) = e^{i k_{i} \cdot x }$ is the plane wave with momentum $k_{i} \in 2\pi \mathbb{Z}^{3}$, compatible with periodic boundary conditions on $\Lambda$, and the sum is over all permutations of $\{1,\ldots, N\}$. The momenta $k_{1},\ldots, k_{N}$ and the number of particles $N$ are chosen so to fill the Fermi ball $\mathcal{B}_{\mu} = \{ k\in 2\pi \mathbb{Z}^{3} \mid \eps^{2} |k|^{2} \leq \mu \}$, up to a fixed chemical potential $\mu = O(1)$.

Let $\omega_{N}$ be the reduced one-particle density matrix of such state:
\begin{eqnarray}
\omega_{N} &:=& N\tr_{2,\ldots, N} |\psi_{\text{Slater}} \rangle \langle \psi_{\text{Slater}}| \nonumber\\
&=& \sum_{i=1}^{N} |f_{k_{i}} \rangle \langle f_{k_{i}}|\;.
\end{eqnarray}
Slater determinants are an example of {\it quasi-free state}, a class of quantum states which are completely determined by their one-particle density matrix. In particular, the energy of $\psi_{\text{Slater}}$ is:
\begin{equation}
\langle \psi_{\text{Slater}}, H_{N} \psi_{\text{Slater}} \rangle = \mathcal{E}_{N}^{\text{HF}}(\omega_{N})
\end{equation}
where $\mathcal{E}_{N}^{\text{HF}}(\cdot)$ is the {\it Hartree-Fock energy functional}:
\begin{eqnarray}\label{eq:HFdef}
 \mathcal{E}^{\text{HF}}_{N}(P_{N}) &=& \tr (- \eps^{2}\Delta P_{N})\\&& + \frac{1}{2N}\int dxdy\, v(x-y) (P_{N}(x;x) P_{N}(y;y) - |P_{N}(x;y)|^{2})\;,\nonumber
\end{eqnarray}
with $0\leq P_{N} \leq 1$, $\tr\, P_{N} = N$. The interaction energy is taken into account by the last two terms in Eq. (\ref{eq:HFdef}), called respectively the direct and the exchange term. The analogue of the Bohm-Pines prediction in the mean-field scaling is:
\begin{eqnarray}\label{eq:BPmf}
E_{N} &=& \mathcal{E}^{\text{HF}}_{N}(\omega_{N}) + \mathcal{C}_{N}\;,\nonumber\\
\mathcal{C}_{N} &=& \eps c_{\text{corr}}^{\text{mf}}(v) + o(\eps)\;,
\end{eqnarray}
for some negative constant $c_{\text{corr}}^{\text{mf}}(v)$, determined by the interaction potential $v$. Notice that $\mathcal{C}_{N}$ is subleading with respect to the smallest term in $\mathcal{E}^{\text{HF}}_{N}(\omega_{N})$, the exchange term, which is $O(1)$ in the mean-field regime, for bounded potentials. 

Let us briefly comment on the notion of correlation energy we used so far. Another definition used in the literature involves the comparison of the many-body ground state energy with the smallest energy reached by an $N$-particle quasi-free state, namely the Hartree-Fock ground state energy:
\begin{equation}
E_{N}^{\text{HF}} = \inf_{\substack{ 0\leq P_{N} \leq 1 \\ \tr\, P_{N} = N}} \mathcal{E}_{N}^{\text{HF}}(P_{N})
\end{equation}
(for positive interaction potentials, the minimum is reached by an orthogonal projection, a fact known as Lieb's variational principle, \cite{L}). In general, the interacting Hartree-Fock ground state is not given by $\omega_{N}$. Nevertheless, one expects the true Hartree-Fock ground state energy to be energetically close to $\mathcal{E}_{N}^{\text{HF}}(\omega_{N})$; this has been recently proven in \cite{GHL} for jellium, with a method that can be extended to the simpler mean-field regime. The proof of \cite{GHL} shows that $E_{N}^{\text{HF}}$ and $\mathcal{E}_{N}^{\text{HF}}(\omega_{N})$ are exponentially close in the number of particles $N$, as $N\gg 1$. This rules out possible ambiguities in the definition of the correlation energy, up to exponentially small error terms.

In this paper we give the first rigorous result on Eq. (\ref{eq:BPmf}). Informally, we prove that, for $v$ regular and $\|v\|_{\infty}$ small enough:
\begin{equation}\label{eq:corrconj2}
E_{N} = \mathcal{E}^{\text{HF}}_{N}(\omega_{N}) + \mathcal{C}^{(2)}_{N} + O(\eps \|v\|^{3}_{\infty})\;,
\end{equation}
where $\mathcal{C}^{(2)}_{N} = O(\eps)$ is an {\it explicit} quantity, see Eqs. (\ref{eq:E2}), (\ref{eq:cormain}), depending quadratically on the interaction potential $v$. The term $\mathcal{C}^{(2)}_{N}$ is precisely what one gets in a formal perturbative expansion of the correlation energy, assuming the interaction potential $v$ to be small. In the case of jellium, this was discussed in the works \cite{BP, Macke}, by using the fine structure constant as a small parameter. In view of the previous comment on the exponential closeness of $\mathcal{E}_{N}^{\text{HF}}(\omega_{N})$ and $E_{N}^{\text{HF}}$, our result provides an expression for the first correction to the ground state energy that is not captured by Hartree-Fock theory.

The proof of Eq. (\ref{eq:corrconj2}) is based on the combination of methods recently introduced for the study of the dynamics of mean-field fermions, see \cite{BPS, BPS2, BJPSS, PRSS}, together with a rigorous formulation of second order perturbation theory, originally developed in the context of non-relativistic QED in \cite{youngchris} and refined in \cite{CH, youngchris2}, see also \cite{HVV, HSei}. 

Let us briefly outline the strategy of the proof. Eq. (\ref{eq:corrconj2}) is proven by matching upper and lower bounds for the ground state energy. The upper bound follows from a suitable choice of the trial state, inspired by second order perturbation theory around the free Fermi gas. Concerning the lower bound, the proof goes as follows. The starting point is Proposition \ref{eq:low1}, that allows to rewrite the many-body energy of any fermionic state as $\mathcal{E}_{N}^{\text{HF}}(\omega_{N})$ plus terms that describe particle excitations around the Fermi level $\mu$. The proof of this proposition is based on the use of fermionic Bogoliubov transformations, adapting ideas developed in the context of many-body quantum dynamics, \cite{BPS}. As later proven in Proposition \ref{prp:boundsQ}, the energetic contribution of the fluctuations around the Fermi level is expressed by a sum of terms that are either positive, or that can be controlled via the total kinetic energy of the excitations. The use of the kinetic energy rather than the number operator to quantify the energetic closeness of the Hartree-Fock ground state to the many-body ground state energy is crucial: due to the absence of a uniform spectral gap, states with comparable energies might differ by the presence of a large number of low energy excitations around the Fermi level.

An immediate application of Proposition \ref{prp:boundsQ}, Corollary \ref{prp:lower}, yields a lower bound for the correlation energy which displays an optimal dependence on the number of particles. In particular, this bound provides a derivation of Hartree-Fock theory for the ground state of fermions interacting via bounded potentials, with an optimal rate of convergence. Despite the optimal $N$-dependence, the bound of Corollary \ref{prp:lower} is off by a multiplicative constant. The situation is then improved in Proposition \ref{cor:Nimpro}, where a better lower bound for the correlation energy is obtained, which matches the prediction of second order perturbation theory. The proof of this proposition is crucially based on Lemma \ref{prp:lower2}, inspired by \cite{youngchris}, that allows to establish the validity of second order perturbation theory as a lower bound for the many-body ground state energy.

To conclude, let us mention that an improved upper bound on the correlation energy, reproducing the analogue of the prediction of \cite{BP, GB} for the mean-field regime at all orders in $v$, has been recently established in \cite{BNPSS}.

The paper is organized as follows. In Section \ref{sec:2nd} we define the model in second quantization. In Section \ref{sec:HF} we state our main result, and in Section \ref{sec:proof} we shall give the proof. Finally, in Appendix \ref{app:A} we collect some explicit computations, and in Appendix \ref{app:interaction} we discuss the details of the proof of Proposition \ref{prp:low1}.

\section{The model}\label{sec:2nd}

We consider a system of $N$ interacting spinless fermions in a cubic box $\Lambda = [0;1]^{3}$ with periodic boundary conditions. The single-particle Hilbert space is $L^{2}(\Lambda)$; an orthonormal basis for such space, compatible with periodic boundary conditions, is given by the plane waves $f_{k}$:
\begin{equation}
f_{k}(x) = e^{ik\cdot x}\;,\qquad k\in 2\pi \mathbb{Z}^{3}\;.
\end{equation} 
We shall describe the system in second quantization, directly in momentum space. The fermionic Fock space is:
\begin{eqnarray}
\mathcal{F} &=& \mathbb{C}\oplus \bigoplus_{n\geq 1} L^{2}_{\text{a}}((2\pi \mathbb{Z}^{3})^{n})\nonumber\\
&=:& \bigoplus_{n\geq 0} \mathcal{F}^{(n)}\;,
\end{eqnarray}
with $\mathcal{F}^{(n)}$ the $n$-particle sector of the Fock space, $\mathcal{F}^{(n)} = L^{2}_{\text{a}}((2\pi \mathbb{Z}^{3})^{n})$ for $n\geq 1$ and $\mathcal{F}^{(0)} = \mathbb{C}$. Thus, any element $\psi$ of the fermionic Fock space $\mathcal{F}$ has the form $\psi = (\psi^{(0)}, \psi^{(1)}, \ldots, \psi^{(n)}, \ldots)$, with $\psi^{(n)}(k_{1}, \ldots, k_{n})\in \mathcal{F}^{(n)}$.  Let us introduce the momentum space fermionic creation and annihilation operators as:
\begin{eqnarray}
(a^{*}_{k} \psi)^{(n)}(k_{1}, \ldots, k_{n}) &=& \frac{1}{\sqrt{n}}\sum_{j=1}^{n} (-1)^{j} \delta_{k, k_{i}} \psi^{(n-1)}(k_{1}, \ldots, k_{i-1}, k_{i+1}, \ldots, k_{n})\;,\nonumber\\
(a_{k} \psi)^{(n)}(k_{1}, \ldots, k_{n}) &=& (\sqrt{n+1}) \psi^{(n+1)}(k, k_{1}, \ldots, k_{n})\;,
\end{eqnarray}
with $\delta_{k, k'}$ the Kronecker delta, and with the understanding that $a_{k} \Omega = 0$, for $\Omega = (1, 0, \ldots, 0, \ldots)$ the vacuum vector. The operators $a^{*}_{k}$ and $a_{k}$ respectively create and annihilate a particle with momentum $k$. They satisfy the canonical anticommutation relations:
\begin{equation}\label{eq:CAR}
\{ a_{k}, a_{k'} \} = \{ a^{*}_{k}, a^{*}_{k'} \} = 0\;,\qquad \{ a^{*}_{k}, a_{k'} \} = \delta_{k,k'}\;.
\end{equation}
The second of Eq. (\ref{eq:CAR}) immediately implies that $\| a_{k} \|\leq 1$. The full Fock space can be generated by repeated applications of creation operators on the vacuum vector $\Omega$. 

Operators acting on the Fock space can be represented in terms of $a_{k}$ and $a^{*}_{k}$. For instance, the {\it number operator} $\mathcal{N}$, acting as $(\mathcal{N} \psi)^{(n)} = n \psi^{(n)}$, is represented as 
\begin{equation}
\mathcal{N} = \sum_{k \in 2\pi \mathbb{Z}^{3}} a^{*}_{k} a_{k}\;.
\end{equation}

%Given a fermionic state $\psi$, we define its reduced $n$-particle density matrix as:
%
%\begin{equation}
%\gamma^{(n)}_{\psi}(k_{1}, \ldots, k_{n}; k'_{n}, \ldots, k'_{1}) = \langle \psi, a^{*}_{k'_{1}}\cdots a^{*}_{k'_{n}}a_{k_{n}}\cdots a_{k_{1}}\psi \rangle\;.
%\end{equation}
%
%Momentum conservation implies that $\gamma^{(n)}_{\psi}(k_{1}, \ldots, k_{n}; k'_{n}, \ldots, k'_{1}) = 0$ unless $\sum_{i=1}^{n} k_{i} = \sum_{i=1}^{n} k'_{i}$.

We are interested in describing a system of nonrelativistic fermions interacting via a bounded two-body potential $v$, compatible with the periodicity of the box $\Lambda$, with Fourier transform $v: 2\pi \mathbb{Z}^{3}\to \mathbb{R}$, $\hat{v}(p) = \hat{v}(-p)$, of positive type: $\hat{v}(p)\geq 0$ for all $p\in 2\pi \mathbb{Z}^{3}$. The second-quantized Hamiltonian of the model is, in momentum space:
\begin{equation}\label{eq:Hfock}
\mathcal{H}_{N} = \sum_{k} \eps^{2} |k|^{2} a^{*}_{k} a_{k} + \frac{1}{2N} \sum_{k, k', p} \hat{v}(p) a^{*}_{k+p} a^{*}_{k'-p} a_{k'} a_{k}\;.
\end{equation}
In the following, all sums will be understood as restricted to $2\pi \mathbb{Z}^{3}$. 

The $N$-particle ground state energy of the system is:
\begin{equation}\label{eq:GS}
E_{N} = \inf_{\psi \in \mathcal{F}^{(N)}} \frac{\langle \psi, \mathcal{H}_{N} \psi\rangle}{\langle \psi, \psi \rangle}\;.
\end{equation}
We shall compare the ground state energy (\ref{eq:GS}) with the energy of the $N$-particle noninteracting ground state (Fermi sea). In order to fix $N$, we shall choose a chemical potential $\mu$. Let $\mathcal{B}_{\mu}$ be the Fermi ball:
\begin{equation}
\mathcal{B}_{\mu} = \{ k\in (2\pi)\mathbb{Z}^{3} \mid \eps^{2} |k|^{2} \leq \mu \}\;,
\end{equation}
where $\mu = O(1)$. We shall set $N \equiv N(\varepsilon, \mu) = |\mathcal{B}_{\mu}|$. It is well known that, for $\varepsilon$ small:
\begin{equation}\label{eq:Neps}
N(\varepsilon, \mu) = \frac{4\pi}{3} \Big(\frac{\mu}{4\pi^2}\Big)^{3/2} \eps^{-3} + o(\varepsilon^{-3})\;.
\end{equation}
For definiteness, in the following we shall choose $\mu$ such that:
\begin{equation}
\frac{4\pi}{3} \Big(\frac{\mu}{4\pi^2}\Big)^{3/2} = 1\;.
\end{equation}
Thus, the only independent parameter is $\varepsilon$. Let $k_{1}, \ldots, k_{N}$ be the momenta contained in $\mathcal{B}_{\mu}$. The noninteracting ground state of the system is obtained by considering the Slater determinant associated to the plane waves $f_{k_{1}}, \ldots, f_{k_{N}}$. The reduced one-particle density matrix of such state is:
\begin{equation}\label{eq:omegaN}
\omega_{N} = \sum_{i = 1}^{N} |f_{k_{i}} \rangle \langle f_{k_{i}}|\;.
\end{equation}
By translation invariance, $\omega_{N}(x;y) \equiv \omega_{N}(x-y)$. In Fourier space:
\begin{eqnarray}\label{eq:omegaHF}
\hat \omega_{N}(k) &=& \int dx\, e^{ik\cdot x}\omega_{N}(x)\nonumber\\
&=& \chi(k\in \mathcal{B}_{\mu})\;.
\end{eqnarray}
The Hartree-Fock energy of (\ref{eq:omegaHF}) is:
\begin{equation}\label{eq:planeHF}
\mathcal{E}^{\text{HF}}_{N}(\omega_{N}) = \sum_{k \in \mathcal{B}_{\mu}} \eps^{2} |k|^{2} + \frac{N \hat{v}(0) }{2} - \frac{1}{2N} \sum_{k, k' \in \mathcal{B}_{\mu}} \hat{v}(k-k') \;.
\end{equation}
Eq. (\ref{eq:planeHF}) provides a very simple, explicit upper bound for the many-body ground state energy. The question we are interested in here is to quantify the error introduced by approximating the ground-state energy by (\ref{eq:planeHF}). To do this, we shall use {\it Bogoliubov transformations}, see {\it e.g.} \cite{So, BLS} for reviews. Given $\omega_{N}$ as in Eq. (\ref{eq:omegaN}), we denote by $R_{\omega_{N}}: \mathcal{F}\to \mathcal{F}$ the implementor of the corresponding Bogoliubov transformation. The implementor $R_{\omega_{N}}$ is a unitary operator on $\mathcal{F}$, that enjoys the following properties:
\begin{eqnarray}\label{eq:propbog}
R_{\omega_{N}} \Omega &=& a^{*}_{k_{1}} a^{*}_{k_{2}}\cdots a^{*}_{k_{N}}\Omega \nonumber\\
R^{*}_{\omega_{N}} a_{k} R_{\omega_{N}} &=& a_{k} \chi(k \in \mathcal{B}^{c}_{\mu}) + a^{*}_{k} \chi(k\in \mathcal{B}_{\mu})\nonumber\\
&\equiv& b_{k} + c^{*}_{k}\;,
\end{eqnarray}
where:
\begin{equation}\label{eq:bkck}
b_{k} := a_{k} \chi(k \in \mathcal{B}^{c}_{\mu})\;,\qquad c_{k} := a_{k} \chi(k\in \mathcal{B}_{\mu})
\end{equation}
and $\mathcal{B}_{\mu}^{c}$ is the complement of the Fermi ball, $\mathcal{B}_{\mu}^{c} = (2\pi) \mathbb{Z}^{3} \setminus \mathcal{B}_{\mu}$. The first property allows to represent the Slater determinant associated to the plane waves $f_{k_{1}}, \ldots, f_{k_{N}}$ in terms of the unitary action of $R_{\omega_{N}}$ on the Fock space vacuum $\Omega$. The second property shows that, after conjugation, the annihilation operator $a_{k}$ becomes a creation operator $a^{*}_{k}$ if $k$ is inside the Fermi ball, or it stays an annihilation operator if $k$ is outside. The state $R_{\omega_{N}} \Omega$ plays the role of ``new vacuum'' for the theory, and it is annihilated by both $b_{k}$ and $c^{*}_{k}$.
%The second property can be understood as a {\it particle-hole conjugation:} before conjugation, operator $a_{k}$ destroys a particle with momentum $k$. After conjugation, this operators splits a sum of two: $b_{k}$ destroys a particle if $k\notin \mathcal{B}_{\mu}$, while $c^{*}_{k}$ creates a particle if $k\in \mathcal{B}_{\mu}$. This last particle is also called a ``hole'' due to the fact that  
These properties will be very useful in comparing the many-body ground state energy with the Hartree-Fock energy.

\section{Main result}\label{sec:HF}
The next theorem provides rigorous upper and lower bounds for the correlation energy of interacting fermions in the mean-field regime. %In the following, we will quantify the strength of the potential via:
%
%\begin{equation}\label{eq:normv}
%\| \hat v \| := \|(1 + |p|^{2})\hat v\|_{1}
%\end{equation}
%
\begin{theorem}[Main result.]\label{thm:1} Let $\hat v: 2\pi\mathbb{Z}^{3}\to \mathbb{R}$, $\hat{v}(p)\geq 0$, $\hat{v}(p) = \hat{v}(-p)$. Let
\begin{equation}
\mathcal{C}_{N} = E_{N} - \mathcal{E}^{\text{HF}}_{N}(\omega_{N})
\end{equation}
be the correlation energy, with $\mathcal{E}^{\text{HF}}_{N}(\omega_{N})$ given by Eq. (\ref{eq:planeHF}) and $N$ chosen as in Eq. (\ref{eq:Neps}). Then, there exist constants $v_{0}, K$, independent of $\eps$, such that for $\|\hat v\|_{1}\leq v_{0}$, $\| |p|^{K}\hat v(p) \|_{1} \leq C_{K}$ the following is true. For any $\delta > 0$ there exists $C_{\delta}>0$ such that:
\begin{equation}\label{eq:corrres}
- C_{\delta} N^{-\frac{9}{16} + \delta} + (1 + C\|\hat v\|_{1}) \mathcal{C}^{(2)}_{N}  \leq \mathcal{C}_{N} \leq  (1 - C \|\hat{v} \|_{1}) \mathcal{C}^{(2)}_{N} + C_{\delta} N^{-\frac{9}{16} + \delta}
\end{equation}
with $\mathcal{E}^{\text{HF}}_{N}(\omega_{N})$ given by Eq. (\ref{eq:planeHF}) and, setting $e(k) := |\eps^{2} |k|^{2} - \mu|$:
\begin{eqnarray}\label{eq:E2}
\mathcal{C}^{(2)}_{N} &=& -\frac{1}{2 N^{2}} \sum_{p}\sum_{\substack{k, k':\, k, k'\in \mathcal{B}_{\mu} \\ k + p, k'-p \in \mathcal{B}_{\mu}^{c}}} \frac{\hat{v}(p)^{2}}{e(k+p) + e(k) + e(k'-p) + e(k')}\nonumber\\
&& + \frac{1}{2N^{2}} \sum_{p}\sum_{\substack{k, k':\, k, k'\in \mathcal{B}_{\mu} \\ k + p, k'-p \in \mathcal{B}_{\mu}^{c}}} \frac{\hat{v}(p) \hat{v}(p - k'+k)}{e(k+p) + e(k) + e(k'-p) + e(k')}\;.
\end{eqnarray}
\end{theorem}
\begin{remark} 
\begin{itemize}
\item[(i)] The expression in Eq. (\ref{eq:E2}) is what one naturally finds in formal perturbative computations of the correlation energy, \cite{Macke, BP, GB}. It turns out that $\mathcal{C}^{(2)}_{N}$ is negative and that $\mathcal{C}^{(2)}_{N} = O(\eps)$. To begin, notice that the last term in Eq. (\ref{eq:E2}) is subleading with respect to the first. This follows from:
\begin{eqnarray}
&&\frac{1}{2N^{2}} \sum_{p}\sum_{\substack{k, k':\, k, k'\in \mathcal{B}_{\mu} \\ k + p, k'-p \in \mathcal{B}_{\mu}^{c}}} \frac{\hat{v}(p) \hat{v}(p - k'+k)}{e(k+p) + e(k) + e(k'-p) + e(k')} \nonumber\\&&\quad \leq \frac{\| \hat v \|_{1}}{2N^{2}} \sum_{p}\sum_{\substack{k:\, k\in \mathcal{B}_{\mu} \\ k + p \in \mathcal{B}_{\mu}^{c}}} \frac{\hat{v}(p)}{e(k+p) + e(k)} \\
&&\quad \leq \frac{\| \hat v \|_{1}}{2N} \sum_{p} \hat v(p) I_{\mu}(p)\;,\nonumber
\end{eqnarray}
where
\begin{equation}
I_{\mu}(p) = \frac{1}{N}\sum_{\substack{k:\, k\in \mathcal{B}_{\mu} \\ k + p \in \mathcal{B}_{\mu}^{c}}} \frac{1}{e(k+p) + e(k)}\;.
\end{equation}
As proven in Lemma \ref{lem:H0}, $I_{\mu}(p) \leq C$ for $|p| \leq N^{\zeta}$ and $\zeta > 0$. Therefore, using also the decay of $\hat v(p)$ for large $|p|$:
\begin{equation}
\frac{1}{2N^{2}} \sum_{p}\sum_{\substack{k, k':\, k, k'\in \mathcal{B}_{\mu} \\ k + p, k'-p \in \mathcal{B}_{\mu}^{c}}} \frac{\hat{v}(p) \hat{v}(p - k'+k)}{e(k+p) + e(k) + e(k'-p) + e(k')} = o(\eps)\;.
\end{equation}
Moreover, an explicit computation, reported in Appendix \ref{app:corr}, shows that:
\begin{equation}\label{eq:cormain}
\lim_{\eps\to 0^{+}} \varepsilon^{-1}\mathcal{C}_{N}^{(2)} = - (2\pi)^{3} \frac{\pi}{2} (1 - \log 2)\sum_{p} |p|\hat{v}(p)^{2}\;.
\end{equation}
Eq. (\ref{eq:cormain}) is consistent with the computations reported in the physics literature, see {\it e.g.} \cite[Eq. (27), (28)]{Macke}. Notice that if one formally replaces the potential $\hat v(p)$ by the Fourier transform of the Coulomb potential, $1/|p|^{2}$, one encounters a well-known logarithmic divergence.

\item[(ii)] The assumption on the decay of $\hat v(p)$ is technical, and far from being optimal. Inspection of the proof shows that $K$ can be chosen equal to $100$.

\item[(iii)] As recently proven in \cite{GHL} for jellium, with a method that can be adapted to the mean-field regime, the difference between the energy of the Free fermi sea $\mathcal{E}^{\text{HF}}_{N}(\omega_{N})$ and the Hartree-Fock ground state energy is exponentially small in $N$. Thus, our theorem allows to compute the next order correction to the many-body ground state energy that is not captured by the Hartree-Fock approximation (notice that this correction is also not captured by Thomas-Fermi nor Dirac-Schwinger theories, since they are all included in the Hartree-Fock approximation).

\item[(iv)] A simple consequence of the method of the proof is the following bound for the many-body reduced density matrix. Let $\psi_{N}$ be a normalized $N$-particle fermionic wave function with energy $E(\psi_{N})$, such that
\begin{equation}
|\mathcal{E}^{\text{HF}}_{N}(\omega_{N}) - E(\psi_{N})| \leq C\varepsilon\;.
\end{equation}
In particular, $\psi_{N}$ could be the fermionic ground state of the model. Let $\gamma_{\psi_{N}} = N \tr_{2,\ldots, N} |\psi_{N}\rangle \langle \psi_{N}|$ be its reduced one particle density matrix. Then, see Remark \ref{rem:delta}, for any $\delta>0$:
\begin{equation}\label{eq:dgo}
\tr\, \gamma_{\psi_{N}} (1 - \omega_{N}) \leq C_{\delta} N^{\frac{7}{16} + \delta}\;.
\end{equation}
The quantity $\delta(\gamma_{\psi_{N}}, \omega_{N}) := \tr\, \gamma_{\psi_{N}} (1 - \omega_{N})$ contains interesting information on the many-body ground state. In particular, it can be used to bound the Hilbert-Schmidt norm $\| \gamma_{\psi_{N}} - \omega_{N} \|_{\text{HS}}$. One has:
\begin{eqnarray}
\| \gamma_{\psi_{N}} - \omega_{N} \|_{\text{HS}}^{2} &=& \tr |\gamma_{\psi_{N}} - \omega_{N}|^{2} \nonumber\\
&\leq& \tr (\gamma_{\psi_{N}} + \omega_{N} - 2\gamma_{\psi_{N}} \omega_{N}) \nonumber\\
&=& 2 \delta(\gamma_{\psi_{N}}, \omega_{N}) \leq 2 C_{\delta} N^{\frac{7}{16} + \delta}\;.
\end{eqnarray}
Here we used that $0\leq \gamma_{\psi_{N}}\leq 1$, $0\leq \omega_{N}\leq 1$, that $\tr\,\gamma_{\psi_{N}} = \tr\,\omega_{N} = N$, and the cyclicity of the trace. The proof leading to Eq. (\ref{eq:dgo}) can be immediately adapted to compare the Hartree-Fock minimizer $\omega_{N}^{\text{gs}}$ with $\omega_{N}$, using that their energy difference is much smaller than $C\varepsilon$, \cite{GHL}. Therefore, one has:
\begin{equation}
\| \gamma_{\psi_{N}} - \omega_{N}^{\text{gs}} \|^{2}_{\text{HS}} \leq 8 C_{\delta} N^{\frac{7}{16} + \delta}\;.
\end{equation} 
\item[(v)] The proof of Theorem \ref{thm:1} makes important use of translation invariance, and of the boundedness of the potential. Nevertheless, we do not expect these assumptions to be essential, and we expect our result to hold in a more general setting. Recently, partial progress in the use of similar methods to derive effective evolution equations for many-body quantum systems interacting via Coulomb potentials has been obtained in \cite{PRSS}. The extension of Theorem \ref{thm:1} to Coulomb potentials would be particularly relevant, to rigorously justify the predictions of \cite{Macke, BP, GB, Sa}.
\item[(vi)] Finally, let us mention that an upper bound on the correlation energy capturing all powers of $v$, for potentials satisfying the same assumptions of Theorem \ref{thm:1}, has been recently proven in \cite{BNPSS}. The result of \cite{BNPSS} yields the analogue of the prediction of \cite{BP, GB} in the mean-field regime, as an upper bound for the correlation energy.
\end{itemize}
\end{remark}

%Theorem \ref{thm:1} immediately implies the following corollary.
%
%\begin{corollary}
%Let the interaction potential be $\lambda v$. Then, under the same assumption of Theorem \ref{thm:1}:
%
%\begin{equation}
%\lim_{\eps \to 0} \frac{\mathcal{C}_{N}}{\lambda^{2} \eps} = c^{(2)} + O(\lambda)\;,
%\end{equation}
%
%with $c^{(2)}$ given by an explicit $O(1)$ constant.
%\end{corollary}
%
%\begin{remark}
%The expression (\ref{eq:E2}) agrees with the second order approximation of the correlation energy, which can be formally computed by naive perturbation theory.
%\end{remark}

The rest of the paper is devoted to the proof of Theorem \ref{thm:1}. In Section \ref{sec:1stlow} we shall prove some important bounds for the many-body interaction, that in particular allow to bound from below the many-body ground state energy in terms of the Hartree-Fock energy, up to an error $O(\varepsilon)$. This result allows to prove the validity of Hartree-Fock theory for the ground state energy, with an optimal rate of convergence. Then, in Section \ref{sec:2ndlow} we shall improve the strategy, and we shall compute the correlation energy at second order in the interaction potential, thus obtaining the lower bound stated in the theorem. Finally, in Section \ref{sec:up} we shall prove a matching upper bound, by choosing a suitable trial state motivated by perturbation theory.

\section{Proof of Theorem \ref{thm:1}}\label{sec:proof}
\subsection{Comparing Hartree-Fock and ground state energies}\label{sec:1stlow}

The starting point is the following proposition, which allows to compare the many-body ground state energy with the Hartree-Fock energy of $\omega_{N}$.
\begin{proposition}\label{prp:low1} %Then, setting:
%
%\begin{eqnarray}
%&&b_{k} := \left\{ \begin{array}{cc} a_{k} & \text{for $k\notin \mathcal{B}_{\mu}$} \\ 0 & \text{otherwise}\end{array} \right.\nonumber\\
%&&c_{k} := \left\{ \begin{array}{cc} a_{-k} & \text{for $k\in \mathcal{B}_{\mu}$} \\ 0 & \text{otherwise,} \end{array}\right.
%\end{eqnarray}
%
The following identity holds true, for all $\varphi \in \mathcal{F}$ such that $\|\varphi\| = 1$ and $\langle \varphi, \mathcal{N} \varphi \rangle = N$:
\begin{equation}\label{eq:low1}
\langle \varphi, \mathcal{H}_{N} \varphi\rangle = \mathcal{E}^{\text{HF}}_{N}(\omega_{N}) + \langle R^*_{\omega_{N}}\varphi, \mathbb{H}_{0} R^*_{\omega_{N}}\varphi \rangle + \langle R^*_{\omega_{N}}\varphi, \mathbb{X} R^*_{\omega_{N}}\varphi \rangle + \langle R_{\omega_{N}}^{*} \varphi, \mathbb{Q} R_{\omega_{N}}^{*} \varphi \rangle
\end{equation}
where, with $b_{k}$, $c_{k}$ defined as in Eq. (\ref{eq:bkck}):
\begin{eqnarray}\label{eq:Q}
&&\mathbb{H}_{0} = \sum_{k} e(k) (b^*_{k} b_{k} + c^{*}_{k} c_{k})\;,\qquad e(k) = |\eps^{2} |k|^{2} - \mu|\;,\nonumber\\
&&\mathbb{X} = - \frac{1}{N} \sum_{p} \hat{v}(p) \sum_{k:\, k \in \mathcal{B}_{\mu}} (b^{*}_{k+p} b_{k+p} - c^{*}_{k+p} c_{k+p})\;,\nonumber\\
&&\mathbb{Q} = \frac{1}{2N} \sum_{p, k, k'} \hat{v}(p) \big\{ b_{k+p}^{*} b^{*}_{k'-p} b_{k'} b_{k} + c^{*}_{k+p} c^{*}_{k' - p} c_{k'} c_{k} \nonumber\\
&& \qquad\qquad\qquad\qquad +2 b^{*}_{k+p} c^{*}_{k} c_{k' - p} b_{k'} - 2 b^{*}_{k+p} c^{*}_{k'} c_{k'-p} b_{k} \\
&&+ \big( b^{*}_{k+p} b^{*}_{k'-p} c^{*}_{k'} c^{*}_{k} - 2b^{*}_{k+p} b^{*}_{k'-p} c^{*}_{k} b_{k'} + 2 b^{*}_{k+p} c^{*}_{k'+p} c^{*}_{k} c_{k'} + h.c. \big)\big\}\;.\nonumber
\end{eqnarray}
\end{proposition}
Let us briefly comment on the terms appearing on the right-hand side of Eq. (\ref{eq:low1}). The term $\langle R^*_{\omega_{N}}\varphi, \mathbb{H}_{0} R^*_{\omega_{N}}\varphi \rangle$ takes into account the absolute value of the difference between the Fermi energy and the kinetic energy of the excitations around the Fermi level. Later, we shall crucially exploit the positivity of this term. The operator $\mathbb{X}$ is the second quantization of the exchange term, while the operator $\mathbb{Q}$ collects all quartic terms. This operator describes pairs of excitations around the Fermi level. 

As we shall see later, the terms contributing to the correlation energy are the kinetic energy, and the $b^{*} b^{*} c^{*} c^{*}$ term (together with its adjoint), describing the creation of two particles inside the Fermi ball and two particles outside the Fermi ball. Equivalently, after conjugation with the Bogoliubov transformation, this operator destroys two particles in the Fermi ball (one also says that it creates two {\it holes}), and creates two particles outside the Fermi ball. All the other terms will be either positive, or will be controlled by the positivity of the kinetic energy.

Let us now prove Proposition \ref{prp:low1}.
\begin{proof} Let $R_{\omega_{N}}$ the implementor of the Bogoliubov transformation associated to $\omega_{N}$, introduced in Eq. (\ref{eq:propbog}). We write, by unitarity of $R_{\omega_{N}}$:
\begin{equation}
\langle \varphi, \mathcal{H}_{N} \varphi\rangle = \langle R^{*}_{\omega_{N}}\varphi, R^{*}_{\omega_{N}}\mathcal{H}_{N} R_{\omega_{N}} R^{*}_{\omega_{N}} \varphi\rangle\;,
\end{equation}
and we compute the conjugated Hamiltonian $R^{*}_{\omega_{N}}\mathcal{H}_{N} R_{\omega_{N}}$ using the particle-hole transformation, Eq. (\ref{eq:propbog}). Let us start with the kinetic energy. We have:
\begin{eqnarray}\label{eq:kinetic}
\sum_{k} \eps^{2} |k|^{2} R^{*}_{\omega_{N}} a^{*}_{k} a_{k} R_{\omega_{N}} &=& \sum_{k} \eps^{2} |k|^{2} (b^{*}_{k} + c_{k})(b_{k} + c^{*}_{k})\nonumber\\
&=& \sum_{k} \eps^{2} |k|^{2} ( b^{*}_{k} b_{k} - c^{*}_{k} c_{k} ) + \sum_{k\in \mathcal{B}_{\mu}} \eps^{2} |k|^{2}\;.
\end{eqnarray}
The last step follows after normal ordering, and using that $b_{k} c_{k} = 0$. Consider now the transformed many-body interaction. Using again the particle-hole transformation rule (\ref{eq:propbog}) and putting the result back into normal order we get, after a straightforward algebra (which is an adaptation of Proposition 3.3 of \cite{BPS}, in the context of quantum dynamics; see also Proposition 3.1 of \cite{BJPSS}):
\begin{eqnarray}\label{eq:interaction}
&&\frac{1}{2N} \sum_{p, k, k'} \hat{v}(p) R^{*}_{\omega_{N}} a^{*}_{k+p} a^{*}_{k'-p} a_{k'} a_{k} R_{\omega_{N}} \\
\qquad &=& \frac{1}{2N} \sum_{p, k, k'} \hat{v}(p) ( b^{*}_{k+p} + c_{k+p} ) (b^{*}_{k'-p} + c_{k'-p}) (b_{k'} + c^{*}_{k'}) (b_{k} + c^{*}_{k})\nonumber\\
\qquad &\equiv& \mathbb{D} + \mathbb{X} + \mathbb{Q} + \frac{N}{2} \hat{v}(0) - \frac{1}{2N}\sum_{k, k'} \hat{v}(k - k')  \widehat \omega_{N}(k) \widehat \omega_{N}(k')\;,\nonumber
\end{eqnarray}
with $\mathbb{Q}$, $\mathbb{X}$ given by Eq. (\ref{eq:Q}), and 
\begin{equation}\label{eq:Direct}
\mathbb{D} = \hat{v}(0) \sum_{k} ( b^{*}_{k} b_{k} - c^{*}_{k} c_{k} )\;.
\end{equation}
For completeness, we report the details of Eq. (\ref{eq:interaction}) in Appendix \ref{app:interaction}. Notice that the last two terms in (\ref{eq:interaction}) plus the last one in (\ref{eq:kinetic}) reproduce the Hartree-Fock energy $\mathcal{E}^{\text{HF}}_{N}(\omega_{N})$. To conclude the proof of Eq. (\ref{eq:low1}), we claim that, for all $\varphi \in \mathcal{F}$ such that $\langle \varphi, \mathcal{N} \varphi \rangle = N$:
\begin{equation}\label{eq:H0def}
\langle R^{*}_{\omega_{N}} \varphi, \sum_{k} ( \eps^{2} |k|^{2} + \hat{v}(0)) ( b^{*}_{k} b_{k} - c^{*}_{k} c_{k} ) R^{*}_{\omega_{N}}\varphi \rangle = \langle R^{*}_{\omega_{N}} \varphi, \mathbb{H}_{0} R^{*}_{\omega_{N}} \varphi \rangle\;.
\end{equation}
To check this, simply notice that:
\begin{equation}
\langle R^{*}_{\omega_{N}} \varphi, \sum_{k} ( b^{*}_{k} b_{k} - c^{*}_{k} c_{k} ) R^{*}_{\omega_{N}}\varphi \rangle = \langle  \varphi, \sum_{k} ( b^{*}_{k} b_{k} + c^{*}_{k} c_{k} ) \varphi \rangle - N = 0\;.
\end{equation}
Therefore,
\begin{eqnarray}
&&\langle R^{*}_{\omega_{N}} \varphi, \sum_{k} ( \eps^{2} |k|^{2} + \hat{v}(0)) ( b^{*}_{k} b_{k} - c^{*}_{k} c_{k} ) R^{*}_{\omega_{N}}\varphi \rangle  \nonumber\\&& \qquad = \langle R^{*}_{\omega_{N}} \varphi, \sum_{k} ( \eps^{2} |k|^{2} - \mu) ( b^{*}_{k} b_{k} - c^{*}_{k} c_{k} ) R^{*}_{\omega_{N}}\varphi \rangle\nonumber\\
&&\qquad \equiv \langle R^{*}_{\omega_{N}} \varphi, \sum_{k} |\eps^{2} |k|^{2} - \mu| ( b^{*}_{k} b_{k} + c^{*}_{k} c_{k} ) R^{*}_{\omega_{N}}\varphi \rangle\;.
\end{eqnarray}
This proves Eq. (\ref{eq:H0def}), and concludes the proof of Proposition \ref{prp:low1}.
\end{proof}
Proposition \ref{prp:low1} is the starting point for the proof of the lower bound of the many-body quantum energy. In the next proposition, we give bounds for the $\mathbb{X}$, $\mathbb{Q}$ terms appearing in Eq. (\ref{eq:low1}), that will play a crucial role in the proof of our main result. To do this, it is convenient to rewrite:
\begin{eqnarray}\label{eq:Qsplit}
\mathbb{Q} &=& \mathbb{Q}_{1} + \mathbb{Q}_{2} + \mathbb{Q}_{3}\nonumber\\
\mathbb{Q}_{1} &:=& \frac{1}{2N} \sum_{p, k, k'} \hat{v}(p) \big( b_{k+p}^{*} b^{*}_{k'-p} b_{k'} b_{k} - 2 b^{*}_{k+p} c^{*}_{k'+p} c_{k'} b_{k} + c^{*}_{k+p} c^{*}_{k'-p} c_{k'} c_{k}\big)\nonumber\\
\mathbb{Q}_{2} &:=& \frac{1}{2N} \sum_{p, k, k'} \hat{v}(p) \big( 2 b^{*}_{k+p} c^{*}_{k} c_{k'-p} b_{k'} + b^{*}_{k+p} b^{*}_{k'-p} c^{*}_{k'} c^{*}_{k} + c_{k} c_{k'} b_{k'-p} b_{k+p}\big)\nonumber\\
&\equiv& \mathbb{Q}_{2,a} + \mathbb{Q}_{2,b} + \mathbb{Q}_{2,c}\;,\nonumber\\
\mathbb{Q}_{3} &:=& \frac{1}{N} \sum_{p, k, k'} \hat{v}(p) \big( - b^{*}_{k+p} b^{*}_{k'-p} c^{*}_{k} b_{k'} +  b^{*}_{k+p} c^{*}_{k'+p} c^{*}_{k} c_{k'}\big) + h.c.
\end{eqnarray}
The second quantization of the exchange term can be easily controlled with the number operator:
\begin{equation}\label{eq:Xbd}
|\langle \varphi, \mathbb{X} \varphi \rangle|\leq \frac{\|\hat{v}\|_{1}}{N} \langle\varphi, \mathcal{N} \varphi \rangle\;.
\end{equation}

\begin{proposition}\label{prp:boundsQ} Under the same assumption on $\hat v$ of Theorem \ref{thm:1} the following is true, for all $\varphi \in \mathcal{F}$ and $\alpha \geq 0$:
\begin{equation}\label{eq:NtoH}
\frac{1}{N}\langle \varphi, \mathcal{N} \varphi \rangle \leq \frac{1}{N} \Big(\sum_{k: e(k) \leq 1/N^{\alpha}} 1\Big)\|\varphi\|^{2} + \frac{1}{N^{1-\alpha}} \langle \varphi, \mathbb{H}_{0}\varphi \rangle\;.
\end{equation}
Also,
\begin{equation}\label{eq:Q1pos}
\langle \varphi, \mathbb{Q}_{1} \varphi \rangle = \frac{1}{2N} \sum_{p} \hat{v}(p) \langle\varphi, D_{p}^{*} D_{p}\varphi \rangle + \langle \varphi, \mathbb{Q}_{1,a} \varphi \rangle
\end{equation}
with $D_{p} = b^{*}_{k+p} b_{k} - c^{*}_{k-p}c_{k}$ and
\begin{equation}\label{eq:Q1abd}
|\langle \varphi, \mathbb{Q}_{1,a} \varphi \rangle| \leq \frac{\|\hat{v}\|_{1}}{N} \langle \varphi, \mathcal{N}\varphi \rangle\;.
\end{equation}
Moreover, for any $\xi>0$:
\begin{eqnarray}\label{eq:Qs}
|\langle \varphi, \mathbb{Q}_{2,a}\varphi \rangle| &\leq& C\| \hat{v}  \|_{1} \langle\varphi, \mathbb{H}_{0}\varphi \rangle + \frak{e}_{K}\\
|\langle \varphi, \mathbb{Q}_{2,\sharp} \varphi \rangle| &\leq& C \varepsilon \| |p| \hat{v} \|_{1} \|\varphi\|^{2} + C\|\hat{v}\|_{1} \langle \varphi, \mathbb{H}_{0}\varphi  \rangle + \frak{e}_{K}\nonumber\\
| \langle \varphi, \mathbb{Q}_{3} \varphi \rangle | &\leq& \frac{\xi}{N} \sum_{p} \hat{v}(p) \langle \varphi, D^{*}_{p} D_{p} \varphi \rangle + \frac{C\| \hat{v}\|_{1}}{\xi} \langle \varphi, \mathbb{H}_{0} \varphi \rangle + \frak{e}_{K}\;,\nonumber
\end{eqnarray}
where: $\sharp = b,c$; $0\leq \frak{e}_{K} \leq C_{K} N^{\frac{1}{3} - \zeta K} \|\varphi\|^{2}$ with $C_{K}, K$ as in Theorem \ref{thm:1} and $\zeta = 1/50$.
\end{proposition}
\begin{remark}
The $\frak{e}_{K}$ error terms are due to the large momenta dependence of $\hat v(p)$. They are smaller than any fixed power of $\eps$, provided $\hat v(p)$ decays fast enough for large $|p|$. As it will be clear in the proof, for $\hat v(p)$ compactly supported one can set $\frak{e}_{K} = 0$.
\end{remark}

The logic of the bound (\ref{eq:NtoH}) is to separate the plane wave states that are energetically too close to the Fermi level, from those that are far from it. The number of particles occupying states that are far from the Fermi level can be estimated via their kinetic energy. Instead, for the states close to the Fermi level one cannot use the kinetic energy; we control the number of particles occupying such states by simply counting the number of available states. This strategy might be improved by defining, from the very beginning, the correlation energy by comparing the many-body ground state energy with the Hartree-Fock ground state energy. The two definitions are equivanent: as recently proven in \cite{GHL} for jellium, the energy of the Fermi sea and the HF ground state energy are exponentially close in the density of particles. The advantage of comparing with the true HF ground state would be that, as proven in \cite{BLLS}, there are no unfilled shells in unrestricted Hartree-Fock theory. In our scaling, one expects the first HF excited state to be separated by an energy gap $O(1/N)$ from the HF ground state. This would imply, in particular, that the analogue of the sum in the left-hand side in Eq. (\ref{eq:numb}) would be empty, for $\alpha = 1$. However, the price to pay in this approach is that one cannot rely on translation invariance anymore, since the HF ground state is not translation invariant. Translation invariance of the reference state plays an important role in our proofs, in particular in the proof of Lemma \ref{lem:H0} below, and even more in the analysis of Section \ref{sec:2ndlow}. 

Before discussing the proof of Proposition \ref{prp:boundsQ}, let us show how these estimates can be used to obtain a lower bound on the correlation energy.
\begin{corollary}\label{prp:lower} Under the same assumptions of Theorem \ref{thm:1}, the following is true. There exists a constant $C>0$, independent of $N$, such that the following bound holds, for all $\varphi \in \mathcal{F}$ such that $\langle \varphi, \mathcal{N}\varphi \rangle = N$ and $\|\varphi\| = 1$:
\begin{equation}\label{eq:lower}
\langle \varphi, \mathcal{H}_{N} \varphi\rangle \geq \mathcal{E}^{\text{HF}}_{N}(\omega_{N}) + (1 - C \| \hat{v} \|_{1}) \langle R^{*}_{\omega_{N}} \varphi, \mathbb{H}_{0} R_{\omega_{N}}^{*} \varphi \rangle - C \||p|\hat{v}\|_{1}\eps - \frak{e}_{K}\;.
\end{equation}
In particular, for $v$ such that $1 - C\|\hat{v}\|_{1}\geq 0$, one has:
\begin{equation}\label{eq:corrlow1}
\mathcal{C}_{N} \geq -C \||p|\hat{v}\|_{1}\eps - \frak{e}_{K}\;.
\end{equation}
\end{corollary}
\begin{remark}\label{rem:corlower}
\begin{itemize}
\item[(i)] This lower bound proves the validity of Hartree-Fock theory, at the level of the ground state energy. With respect to previous results, \cite{Ba, Ba2, GS}, we only consider bounded potentials, and we assume translation invariance. The improvement with respect to the preexisting literature is that, as we shall see, our bound on the error term is optimal in its $\varepsilon$-dependence. We expect to be possible to extend Proposition \ref{prp:lower} to a larger class of nonnegative potentials (in particular, without any a priori smallness condition), and to the nontranslation invariant setting; we defer these extensions to future work.

In Section \ref{sec:2ndlow} we shall improve the constant appearing in the lower bound (\ref{eq:corrlow1}), by showing that it can be replaced by an expression that matches second order perturbation theory for small potential. %Combined with the opper bound of Section \ref{sec:up}, this shows in particular the optimality in $\varepsilon$ of the bound (\ref{eq:corrlow1}).

\item[(ii)] Suppose that $|\mathcal{E}^{\text{HF}}_{N}(\omega_{N}) - \langle \varphi, \mathcal{H}_{N} \varphi \rangle| \leq C\varepsilon$. Then, Eq. (\ref{eq:lower}) immediately implies that, for $\hat v$ small enough:
\begin{equation}\label{eq:kinbd}
\langle R^{*}_{\omega_{N}} \varphi, \mathbb{H}_{0} R_{\omega_{N}}^{*} \varphi \rangle \leq C\varepsilon\;.
\end{equation}
\end{itemize}
\end{remark}
\begin{proof}(of Corollary \ref{prp:lower}.) The proof is a direct consequence of Proposition \ref{prp:boundsQ}. Let $\alpha = 1$ in Eq. (\ref{eq:NtoH}). We have:
\begin{equation}\label{eq:numb}
\frac{1}{N} \Big(\sum_{k:e(k) \leq 1/N} 1\Big) \equiv \frac{1}{N} \Big( \sum_{\substack{k\in 2\pi \mathbb{Z}^{3} \\ | \eps^{2} |k|^{2} - \mu | \leq 1/N }} 1 \Big) \leq C\varepsilon.
\end{equation} 
Therefore, Eq. (\ref{eq:low1}) together with Eq. (\ref{eq:numb}) and the bounds of Proposition \ref{prp:boundsQ} easily implies, for some $\delta$-dependent constant $C_{\delta}$:
\begin{eqnarray}
\langle \varphi, \mathcal{H}_{N} \varphi\rangle &\geq& \mathcal{E}^{\text{HF}}_{N}(\omega_{N}) + (1 - C_{\delta}\|\hat{v}\|_{1}) \langle \varphi, \mathbb{H}_{0} \varphi \rangle\\
&& + \frac{1}{2N}(1 - 2\delta) \sum_{p} \hat{v}(p) \langle \varphi, D_{p}^{*} D_{p}\varphi \rangle - C\varepsilon \| |p| \hat{v} \|_{1} - \frak{e}_{K}\;,\nonumber
\end{eqnarray}
which gives Eq. (\ref{eq:lower}), using that $\hat{v}(p)\geq 0$ and choosing, {\it e.g.}, $\delta = 1/4$.
\end{proof}
\begin{remark}\label{rem:impro} 
As the proof suggests, the $O(\varepsilon)$ terms contributing to the lower bound for the correlation energy come from the estimate of the number operator in terms of the kinetic energy, see Eqs. (\ref{eq:NtoH}), (\ref{eq:numb}), and from the first term appearing in the bound of the $\mathbb{Q}_{2,b}$, $\mathbb{Q}_{2,c}$ terms, see the second of Eqs. (\ref{eq:Qs}). The improved bound of Section \ref{sec:2ndlow} will be based on a better estimate for the sum in the left-hand side of Eq. (\ref{eq:numb}), and on the use of the kinetic energy operator $\mathbb{H}_{0}$ to partially control the $\mathbb{Q}_{2,b}$, $\mathbb{Q}_{2,c}$ terms.
\end{remark}

The proof of Proposition \ref{prp:boundsQ} is based on the next key lemma, whose proof is deferred to Appendix \ref{app:key}.
\begin{lemma}\label{lem:H0} Let $\varphi \in \mathcal{F}$. Then, there exists constants $C>0$, $\zeta>0$ independent of $N$ such that, for $p\in 2\pi \mathbb{Z}^{3}$, $|p| \leq N^{\zeta}$:
\begin{eqnarray}\label{eq:estH0}
&&\frac{1}{N^{\frac{1}{2}}}\sum_{k} \big\| b_{k+p} c_{k} \varphi \big\| \leq C I_{\mu}(p)^{\frac{1}{2}}\| \mathbb{H}_{0}^{\frac{1}{2}}\varphi \|\;,\nonumber\\
&&I_{\mu}(p) = \frac{1}{N} \sum_{k:\, k+p\notin \mathcal{B}_{\mu},\, k\in \mathcal{B}_{\mu}} \frac{1}{e(k+p) + e(k)}\;,\qquad I_{\mu}(p) \leq C\;.
\end{eqnarray}
\end{lemma}
\begin{remark} Inspection of the proof shows that one can choose $\zeta = 1/50$.
\end{remark}

%\begin{remark} 
%\begin{itemize}
%\item[(i)] The bounds (\ref{eq:estH0}), together with the boundedness of the fermionic operators, imply:
%
%\begin{eqnarray}\label{eq:estbc}
%\frac{1}{N^{\frac{1}{2}}} \sum_{k} \| b_{k+p} c_{k} \varphi \| &=& \frac{1}{N^{\frac{1}{2}}} \sum_{k: e(k) + e(k+p) \leq A|p|^{2}\eps^{2}} \| b_{k+p} c_{k} \varphi \|\\&& + \frac{1}{N^{\frac{1}{2}}} \sum_{k: e(k) + e(k+p) > A|p|^{2}\eps^{2}} \| b_{k+p} c_{k} \varphi \|\nonumber\\
%&\leq&  C_{\delta} A |p|^{4} N^{\frac{1}{3}(\gamma + \delta) - \frac{1}{2}} + C |p|^{\frac{1}{2}}\| \mathbb{H}_{0}^{1/2} \varphi \|\;.\nonumber
%\end{eqnarray}
%
%\item[(ii)] Notice that the first bound in Eq. (\ref{eq:estH0}) is of order $\eps^{\frac{1}{2} + 1 - \gamma - \frac{\delta}{3}} \equiv \eps^{\frac{1}{2} + \frac{\xi}{2}}$, which is subleading with respect to $\eps^{\frac{1}{2}}$ if $\xi>0$ ($\varepsilon$ is going to be the size of the expectation of $\mathbb{H}_{0}$ on vectors that are close enough to the ground state).
%
%\item[(iii)] Inspection of the proof shows that one can pick, {\it e.g.}, $\delta_{0} = \frac{1}{4}$, for which $\gamma + \delta < 1$, that is $\xi >0$.
%\end{itemize}
%\end{remark}

We are now ready to prove Proposition \ref{prp:boundsQ}.
\begin{proof}(of Proposition \ref{prp:boundsQ}.) Let us start by proving Eq. (\ref{eq:NtoH}). We have:
\begin{eqnarray}\label{eq:calN}
&&\frac{1}{N} \langle \varphi, \mathcal{N} \varphi  \rangle =  \frac{1}{N} \sum_{k} \langle \varphi, a^{*}_{k}a_{k} \varphi \rangle \nonumber\\&&\quad = \frac{1}{N} \Big[\sum_{k:\, e(k) \leq 1/N^{\alpha}} \langle \varphi, a^{*}_{k} a_{k} \varphi \rangle \nonumber + \sum_{k:\, e(k) > 1/N^{\alpha}} \langle \varphi, a^{*}_{k} a_{k} \varphi \rangle \Big]\nonumber\\
&&\quad \leq \frac{1}{N}\Big(\sum_{k:\, e(k) \leq 1/N^{\alpha}} 1\Big) + \frac{1}{N^{1-\alpha}} \langle \varphi, \mathbb{H}_{0} \varphi\rangle\;,
\end{eqnarray}
where in the last step we used that:
\begin{equation}
\frac{1}{N}\sum_{k: e(k) > 1/N^{\alpha}} a^{*}_{k} a_{k} \leq \frac{1}{N^{1-\alpha}}\sum_{k} e(k) a^{*}_{k} a_{k} \equiv \frac{1}{N^{1-\alpha}}\mathbb{H}_{0}\;.
\end{equation}
This proves Eq. (\ref{eq:NtoH}). Let us now consider the quartic terms. Let us start by proving Eq. (\ref{eq:Q1pos}). We rewrite $\mathbb{Q}_{1}$ as:
\begin{eqnarray}
\mathbb{Q}_{1} &=& \frac{1}{2N}\sum_{p, k, k'} \hat{v}(p)\big(b^{*}_{k+p} b_{k} b^{*}_{k'-p} b_{k'} - b^{*}_{k+p} b_{k} c^{*}_{k'+p} c_{k'}\nonumber\\&& - c^{*}_{k'+p} c_{k'} b^{*}_{k+p} b_{k} + c^{*}_{k+p} c_{k} c^{*}_{k'-p} c_{k'}  \big) + \mathbb{Q}_{1,a}\;,
\end{eqnarray}
with:
\begin{equation}
\mathbb{Q}_{1,a} = \frac{1}{2N} \sum_{p, k, k'} \hat{v}(p) \big( - b^{*}_{k+p} b_{k'} \delta_{k'-p,k}\chi(k\notin \mathcal{B}_{\mu}) - c^{*}_{k+p} c_{k'} \delta_{k'-p,k}\chi(k\in \mathcal{B}_{\mu}) \big)\;.
\end{equation}
We further rewrite $\mathbb{Q}_{1}$ as, using that $\hat{v}(p) = \hat{v}(-p)$:
\begin{eqnarray}\label{eq:Q10}
\mathbb{Q}_{1} &=& \frac{1}{2N}\sum_{p, k, k'} \hat{v}(p)\big(b^{*}_{k+p} b_{k} b^{*}_{k'-p} b_{k'} - b^{*}_{k+p} b_{k} c^{*}_{k'+p} c_{k'}\nonumber\\&& - c^{*}_{k'-p} c_{k'} b^{*}_{k-p} b_{k} + c^{*}_{k-p} c_{k} c^{*}_{k'+p} c_{k'}  \big) + \mathbb{Q}_{1,a}\nonumber\\
&\equiv& \frac{1}{2N} \sum_{p} \hat{v}(p) D^{*}_{p} D_{p} + \mathbb{Q}_{1,a}\;,
\end{eqnarray}
with $D_{p} = \sum_{k} (b^{*}_{k+p} b_{k} - c^{*}_{k-p} c_{k})$, $D^{*}_{p} = D_{-p}$. Notice that, thanks to $\hat{v}(p) \geq 0$, the first term in the last line of Eq. (\ref{eq:Q10}) is positive. Consider the term $\mathbb{Q}_{1,a}$. We estimate it as, for any $\varphi\in \mathcal{F}$:
\begin{eqnarray}
|\langle \varphi, \mathbb{Q}_{1,a} \varphi\rangle| &\leq& \frac{1}{2N}\sum_{p} \hat{v}(p) \sum_{k} \big[ \| b_{k+p}\varphi \|^{2} + \| c_{k+p}\varphi \|^{2}\big]\\
&\leq& \frac{\|\hat{v}_{1}\|}{N}  \langle \varphi, \mathcal{N} \varphi \rangle\;.\nonumber
\end{eqnarray}
This proves Eqs. (\ref{eq:Q1pos}), (\ref{eq:Q1abd}). Let us consider the term $\mathbb{Q}_{2}$. We shall estimate its three contributions separately. We have:
\begin{equation}\label{eq:Q2a}
|\langle \varphi, \mathbb{Q}_{2,a} \varphi \rangle| \leq \frac{1}{N}\sum_{p} \hat{v}(p) \sum_{k, k'} \big\| b_{k+p} c_{k} \varphi \big\| \big\| b_{k'-p} c_{k'} \varphi  \big\|\;.
%&\leq& C\|\hat{v}\|_{1} \langle \varphi, \mathbb{H}_{0} \varphi \rangle + C_{\delta}\eps^{1 + \xi}\;,
\end{equation}
We rewrite $\hat v(p)$ as $\hat v_{<}(p) + \hat v_{>}(p)$, with $\hat v_{>}(p)$ supported for momenta $|p| > N^{\zeta}$, with $\zeta$ as in Lemma \ref{lem:H0}. The term corresponding to $v_{>}(p)$ is estimated as:
\begin{eqnarray}
\frac{1}{N}\sum_{p} \hat{v}_{>}(p) \sum_{k, k'} \big\| b_{k+p} c_{k} \varphi \big\| \big\| b_{k'-p} c_{k'} \varphi  \big\| &\leq& CN\eps^{2} \sum_{p} \hat v_{>}(p)|p|^{2}\nonumber\\ &\leq& C_{K} N^{\frac{1}{3} - \zeta K}\;,
\end{eqnarray}
where we used that, by assumption on the potential, $\| \hat v(p) |p|^{K} \|_{1} \leq C_{K}$, together with $\|b_{k}\| \leq 1$, $\|c_{k}\| \leq 1$ and with:
\begin{equation}\label{eq:countnnopt}
\sum_{k:\, k+p \in \mathcal{B}_{\mu}^{c}, \, k\in \mathcal{B}_{\mu}} 1 \leq CN\eps |p|\;.
\end{equation}
Consider now the contribution of $\hat v_{<}(p)$. We have, by Lemma \ref{lem:H0}:
\begin{equation}
\frac{1}{N} \sum_{p} \hat v_{<}(p) \sum_{k, k'} \big\| b_{k+p} c_{k} \varphi \big\| \big\| b_{k'-p} c_{k'} \varphi  \big\| \leq C\| \hat v  \|_{1} \langle \varphi, \mathbb{H}_{0} \varphi \rangle\;.
\end{equation}
This proves the first of Eq. (\ref{eq:Qs}). Now, consider the term $\mathbb{Q}_{2,b}$. We have:
\begin{eqnarray}\label{eq:Q2b1}
&&|\langle \varphi, \mathbb{Q}_{2,b} \varphi \rangle| \leq \frac{1}{2N} \sum_{p} \hat{v}(p) \Big| \Big\langle \sum_{k} b^{*}_{k+p} c^{*}_{k} \varphi, \sum_{k'} b_{k' - p} c_{k'} \varphi \Big\rangle \Big|\\
&& \leq \frac{1}{4N} \sum_{p} \hat{v}_{<}(p) \Big[ \Big\| \sum_{k} b^{*}_{k+p} c^{*}_{k} \varphi \Big\|^{2} + \Big\|  \sum_{k'} b_{k' - p} c_{k'} \varphi  \Big\|^{2} \Big] + C_{K} N^{\frac{1}{3} - \zeta K}\nonumber\\
&& \leq \frac{1}{4N} \sum_{p} \hat{v}_{<}(p) \Big[ \Big\| \sum_{k} b^{*}_{k+p} c^{*}_{k} \varphi \Big\|^{2} + C N \langle \varphi, \mathbb{H}_{0} \varphi \rangle \Big] + C_{K} N^{\frac{1}{3} - \zeta K}\;,\nonumber
\end{eqnarray}
where in the last step we used the triangle inequality and Lemma \ref{lem:H0}. Then, we write:
\begin{eqnarray}\label{eq:Q2b2}
&&\Big\| \sum_{k} b^{*}_{k+p} c^{*}_{k} \varphi \Big\|^{2} = \sum_{k,k'} \langle \varphi, c_{k'} b_{k'+p} b^{*}_{k+p} c^{*}_{k} \varphi \rangle\nonumber\\
&&= \sum_{k,k'} \langle \varphi, b^{*}_{k+p} c^{*}_{k} c_{k'} b_{k'+p} \varphi \rangle + \sum_{k: k+p\in \mathcal{B}_{\mu}^{c}} \langle \varphi, c_{k} c^{*}_{k} \varphi \rangle - \sum_{k: k\in \mathcal{B}_{\mu}} \langle \varphi, b^{*}_{k+p} b_{k+p} \varphi \rangle \nonumber\\
&&\leq \Big\| \sum_{k} b_{k+p} c_{k}\varphi  \Big\|^{2} + CN\eps |p|\;.
\end{eqnarray}
Using again Lemma \ref{lem:H0}, we get:
\begin{eqnarray}
|\langle \varphi, \mathbb{Q}_{2,b} \varphi \rangle| &\leq& \frac{1}{4N}\sum_{p} \hat{v}_{<}(p) \Big[ CN\eps|p| + C N \langle \varphi, \mathbb{H}_{0} \varphi \rangle\Big] + C_{K} N^{\frac{1}{3} - \zeta K} \nonumber\\
&\leq& C\||p|\hat v\|_{1} \eps  + C\| \hat v \|_{1} \langle \varphi, \mathbb{H}_{0} \varphi \rangle + C_{K} N^{\frac{1}{3} - \zeta K}\;.
\end{eqnarray}
The term $\mathbb{Q}_{2,c}$ is estimated in exactly the same way. This proves the second of Eq. (\ref{eq:Qs}). Finally, consider the term $\mathbb{Q}_{3}$. We rewrite is as:
\begin{eqnarray}
\mathbb{Q}_{3} &=& \frac{1}{N} \sum_{p,k,k'} \hat{v}(p) ( b^{*}_{k+p} c^{*}_{k} b^{*}_{k'-p} b_{k'} - b^{*}_{k+p} c^{*}_{k} c^{*}_{k'+p} c_{k'} ) + \text{h.c} \nonumber\\
&\equiv& \frac{1}{N} \sum_{k,p} \hat{v}(p) b^{*}_{k+p} c^{*}_{k} D_{-p} + \text{h.c..}
\end{eqnarray}
By Cauchy-Schwarz inequality, for any $\xi > 0$ (recalling that $D^{*}_{p} = D_{-p}$):
\begin{eqnarray}
| \langle \varphi, \mathbb{Q}_{3} \varphi \rangle | &\leq& \frac{2}{N} \sum_{p} \hat{v}(p) \Big\| \sum_{k} b_{k+p} c_{k} \varphi \Big\| \Big\| D_{-p} \varphi \Big\|\nonumber\\
&\leq& \frac{\xi}{N} \sum_{p} \hat{v}(p) \langle \varphi, D^{*}_{p} D_{p} \varphi \rangle + \frac{1}{\xi N} \sum_{p} \hat{v}(p) \Big\| \sum_{k} b_{k+p} c_{k} \varphi \Big\|^{2}\;.\nonumber
\end{eqnarray}
Therefore, by Lemma \ref{lem:H0} we get:
\begin{eqnarray}\label{eq:Q3}
| \langle \varphi, \mathbb{Q}_{3} \varphi \rangle | &\leq& \frac{\xi}{N} \sum_{p} \hat{v}(p) \langle \varphi, D^{*}_{p} D_{p} \varphi \rangle + \frac{C}{\xi} \sum_{p} \hat{v}(p) \langle \varphi, \mathbb{H}_{0} \varphi \rangle\nonumber\\ && + C_{K} N^{\frac{1}{3} - \zeta K}\;.
\end{eqnarray}
This proves the last of Eqs. (\ref{eq:Qs}), and concludes the proof of Proposition \ref{prp:boundsQ}.
%In conclusion, choosing $\delta = 1/4$, putting together Eqs. (\ref{eq:Q1}), (\ref{eq:Q2}), (\ref{eq:Q3}), we find:
%
%\begin{eqnarray}\label{eq:low2}
%\langle \varphi, \mathbb{Q} \varphi \rangle &\geq& \langle \varphi, \mathbb{Q}_{1} \varphi \rangle - | \langle \varphi, \mathbb{Q}_{2} \varphi \rangle | - |\langle \varphi, \mathbb{Q}_{3} \varphi \rangle| \\
%&\geq& \frac{1}{4N} \sum_{p} \hat{v}(p) \langle \varphi, D^{*}_{p} D_{p}\varphi \rangle - C\| \hat{v} \|_{1} \langle \varphi, %\mathbb{H}_{0} \varphi\rangle - C\|(1+|p|)V\|_{1}\eps\;.\nonumber
%\end{eqnarray}
%
%The final claim follows from the combination of the bounds (\ref{eq:low12}), (\ref{eq:low2}).
\end{proof}
\subsection{Proof of Theorem \ref{thm:1}: lower bound}\label{sec:2ndlow}

In this section we shall prove the lower bound for the correlation energy announced in Theorem \ref{thm:1}. As stressed in Remark \ref{rem:impro}, the lower bound for the correlation energy obtained in Eq. (\ref{eq:corrlow1}) is determined by the estimates of the terms $\mathbb{Q}_{2,b}$, $\mathbb{Q}_{2,c}$ in Eqs. (\ref{eq:Qs}), and from the estimate of the counting of lattice points in Eq. (\ref{eq:numb}). As the next proposition shows, this last point can be improved by using more refined estimates on the Gauss' sphere problem.
\begin{proposition}\label{prp:number} For every $\delta>0$ there exists a constant $K_{\delta}>0$, independent of $N$, such that for every $\alpha>0$ the following bound holds true:
\begin{equation}\label{eq:countbetter}
\frac{1}{N}\Big(\sum_{k:\, e(k) \leq 1/N^{\alpha}} 1\Big) \leq C N^{-\alpha} + K_{\delta} N^{-\frac{9}{16} + \delta} \;,\qquad \text{($e(k) = |\eps^{2} |k|^{2} - \mu|$).}
\end{equation}
\end{proposition}
\begin{proof} The proof is based on refined estimates on the counting of lattice points. As proved in \cite{HB}, see also \cite{Iwa}, for any $R>0$, $\delta>0$ there exists $C_{\delta}>0$ independent of $R$ such that:
\begin{eqnarray}\label{eq:count}
\sum_{k\in \mathbb{Z}^{3}: |k| \leq R} 1 &=& \frac{4\pi}{3} R^{3} + E_{3}(R) \nonumber\\
|E_{3}(R)| &\leq& C_{\delta} R^{\frac{21}{16} + \delta}\;.
\end{eqnarray}
Therefore, we easily get:
\begin{eqnarray}\label{eq:numbth}
\frac{1}{N} \Big(\sum_{k:\, e(k) \leq 1/N^{\alpha}} 1\Big) &=& \frac{1}{N} \Big(\sum_{k:\, ||k|^{2} - N^{2/3}\mu| \leq N^{2/3 - \alpha}} 1\Big)\nonumber\\
&\leq& C N^{-\alpha} + C C_{\delta} N^{-\frac{9}{16} + \frac{\delta}{3}}\;.
\end{eqnarray}
This concludes the proof.
\end{proof}
\begin{remark} The nonoptimal bound (\ref{eq:countnnopt}) is implied by Eq. (\ref{eq:count}) by using the weaker bound $|E_{3}(R)|\leq CR^{2}$. In general, it is enough to know that $E_{3}(R) = o(R^{2})$ to prove that the left-hand side of Eq. (\ref{eq:countbetter}) is $o(\varepsilon)$, which is what we need to compute the correlation energy.% (which is what we need in order to capture the correlation energy at $O(\varepsilon)$).
\end{remark}
This proposition, combined with the estimate (\ref{eq:NtoH}) for the number operator, immediately implies the following corollary.
\begin{corollary}\label{cor:Nimpro} The following improved estimate for the number operator holds true:
\begin{equation}\label{eq:Nimpro}
\langle \xi, \mathcal{N} \xi \rangle \leq CN^{1-\alpha}\|\xi\|^{2} + K_{\delta} N^{\frac{7}{16} + \delta}\|\xi\|^{2} + N^{\alpha}\langle \xi, \mathbb{H}_{0}\xi \rangle\;,\qquad \forall \xi \in \mathcal{F}\;.
\end{equation}
\end{corollary}
Therefore, if $\|\xi\| = 1$ and $\langle \xi, \mathbb{H}_{0}\xi \rangle\leq C \varepsilon$, optimizing over $\alpha$:
\begin{equation}
\langle \xi, \mathcal{N} \xi \rangle \leq C_{\delta} N^{\frac{7}{16} + \delta}\;.
\end{equation}
\begin{remark}\label{rem:delta}
Suppose that $\xi = R^{*}_{\omega_{N}} \varphi$, with $\varphi$ a normalized $N$-particle state, $\varphi = (0, 0, \ldots, 0, \psi_{N}, 0,\ldots)$. Then,
\begin{eqnarray}
\langle R^{*}_{\omega_{N}} \varphi, \mathcal{N} R^{*}_{\omega_{N}} \varphi \rangle &=& \langle \varphi, \mathcal{N} \varphi\rangle + N  - 2\big\langle \varphi, \sum_{k} c^{*}_{k} c_{k} \varphi \big\rangle \nonumber\\
&=& 2N - 2 \tr\, \gamma_{\psi_{N}} \omega_{N} \equiv 2\tr\, \gamma_{\psi_{N}} (1 - \omega_{N})\;.
\end{eqnarray}
Suppose that $| \langle \varphi, \mathcal{H}_{N} \varphi \rangle - \mathcal{E}^{\text{HF}}_{N}(\omega_{N})| \leq C\varepsilon$. Then, Eq. (\ref{eq:Nimpro}) together with the bound for the kinetic energy (\ref{eq:kinbd}) implies:
\begin{equation}
\tr\, \gamma_{\psi_{N}} (1 - \omega_{N}) \leq C_{\delta} N^{\frac{7}{16} + \delta}\;.
\end{equation}
which proves Eq. (\ref{eq:dgo}).
\end{remark}

Therefore, the only negative $O(\varepsilon)$ term contributing to the lower bound for the correlation energy is produced by the estimate of the term
\begin{equation}\label{eq:defF}
\mathbb{Q}_{2,c} = \frac{1}{2N} \sum_{p, k, k'} \hat v(p) c_{k} c_{k'} b_{k'-p} b_{k+p} =: \mathbb{F}\;,
\end{equation}
that is by the creation of two excitations outside the Fermi ball, and two particles (also called holes) inside the Fermi ball; recall Eqs. (\ref{eq:Q2b1}), (\ref{eq:Q2b2}). The goal of the present section will be to improve this bound, by partially controlling this term using some of the kinetic energy $\mathbb{H}_{0}$. As a result, we will get a lower bound for the correlation energy, that agrees with second order perturbation theory. The next proposition proves the lower bound in Eq. (\ref{eq:corrres}) of Theorem \ref{thm:1}. 
\begin{proposition}[Lower bound for the correlation energy]\label{thm:cor} Under the assumptions of Theorem \ref{thm:1}, the following is true. For all $\psi\in \mathcal{F}$ such that $\langle \psi, \mathcal{N} \psi \rangle = N$ and $\|\psi\| = 1$, for all $\delta>0$:
\begin{eqnarray}\label{eq:corrlow}
\langle \psi, \mathcal{H}_{N} \psi\rangle &\geq& \mathcal{E}^{\text{HF}}_{N}(\omega_{N}) + C\| \hat{v} \|_{1} \langle R^{*}_{\omega_{N}} \psi, \mathbb{H}_{0} R_{\omega_{N}}^{*} \psi \rangle\nonumber\\&& - \frac{1}{1 - C\| \hat{v} \|_{1}} \langle \Omega, \mathbb{F} \mathbb{H}_{0}^{-1} \mathbb{F}^{*} \Omega \rangle - C_{\delta} N^{-\frac{9}{16} + \delta}\;,
\end{eqnarray}
where:
\begin{eqnarray}\label{eq:FHF}
\langle \Omega, \mathbb{F} \mathbb{H}_{0}^{-1} \mathbb{F}^{*} \Omega \rangle &=& \frac{1}{2 N^{2}} \sum_{p}\sum_{\substack{k, k':\, k, k'\in \mathcal{B}_{\mu} \\ k + p, k'-p \in \mathcal{B}_{\mu}^{c}}} \frac{\hat{v}(p)^{2}}{e(k+p) + e(k) + e(k'-p) + e(k')}\nonumber\\
&& - \frac{1}{2N^{2}} \sum_{p}\sum_{\substack{k, k':\, k, k'\in \mathcal{B}_{\mu} \\ k + p, k'-p \in \mathcal{B}_{\mu}^{c}}} \frac{\hat{v}(p) \hat{v}(p - k'+k)}{e(k+p) + e(k) + e(k'-p) + e(k')}\;.\nonumber\\
\end{eqnarray}
\end{proposition}
\begin{remark}\label{rem:welldef}
The operator $\mathbb{H}_{0}^{-1}$ will only be defined on vectors $\varphi \in \mathbb{F}^{*} \mathcal{F}$, recall Eq. (\ref{eq:defF}). Clearly, $\varphi^{(n)} = 0$ for $n<4$. We define:
\begin{eqnarray}\label{eq:H-1}
(\mathbb{H}_{0}^{-1}\varphi)^{(n)}(k_{1}, \ldots, k_{n}) &:=& \frac{1}{\sum_{i=1}^{n} e(k_{i})} \varphi^{(n)}(k_{1}, \ldots, k_{n})\;,\quad n\geq 4 \nonumber\\
\qquad\qquad (\mathbb{H}_{0}^{-1}\varphi)^{(n)} &:=& 0\;,\qquad n<4\;.
\end{eqnarray}
Notice that $\varphi^{(n)}(k_{1},\ldots, k_{n})$ is supported on momenta such that $\sum_{i=1}^{n} e(k_{i}) > 0$. This follows from the fact that the particles created by the $b$ operators in $\mathbb{F}^{*}$ have energies strictly larger than the Fermi level $\mu$. The right-hand side of the first of Eq. (\ref{eq:H-1}) is understood as zero whenever $k_{1}, \ldots. k_{n}$ are outside the support of $\varphi^{(n)}(k_{1}, \ldots, k_{n})$.
\end{remark}
The proof of Proposition \ref{thm:cor} relies on the bounds of Proposition \ref{prp:boundsQ}, and on the following lemma, whose proof is inspired by a strategy used in the context of non-relativistic QED \cite[Sec. 3.2]{youngchris}, cf. \cite[Sec. 2.2]{CH} and \cite[Appendix A]{youngchris2}. 
\begin{lemma}\label{prp:lower2} Under the same assumptions of Proposition \ref{thm:cor}, the following inequality holds true, for all $\varphi \in \mathcal{F}$ such that $\|\varphi\| = 1$ and for all $\alpha > 0$:
\begin{eqnarray}\label{eq:youngchris}
\langle \varphi, \alpha \mathbb{H}_{0} \varphi \rangle + \langle \varphi, (\mathbb{F} + \mathbb{F}^{*}) \varphi \rangle &\geq& - \frac{1}{\alpha} \langle \Omega, \mathbb{F} \mathbb{H}_{0}^{-1} \mathbb{F}^{*} \Omega \rangle - \frac{C\| \hat{v} \|_{1}^{2}}{\alpha} \langle \varphi, \mathbb{H}_{0} \varphi \rangle\nonumber\\ &&- C_{\delta} N^{-\frac{9}{16} + \delta}\;.
\end{eqnarray}
\end{lemma}
Before discussing the proof of Lemma \ref{prp:lower2}, let us show how it can be used to prove Proposition \ref{thm:cor}.

\begin{proof}(of Proposition \ref{thm:cor}.) 
Consider the operators $\mathbb{Q}_{1}$, $\mathbb{Q}_{2,a}$, $\mathbb{Q}_{3}$. By Proposition \ref{prp:boundsQ}, we have:
\begin{equation}
\langle \varphi, \widetilde{\mathbb{Q}} \varphi \rangle \geq \frac{1}{4N} \sum_{p} \hat{v}(p) \langle \varphi, D^{*}_{p} D_{p}\varphi \rangle - C\| \hat{v} \|_{1} \langle \varphi, \mathbb{H}_{0} \varphi\rangle - C_{\delta} N^{-\frac{9}{16} + \delta} \;.
\end{equation}
This bound together with the identity Eq. (\ref{eq:low1}) implies:
\begin{eqnarray}
\langle \psi, \mathcal{H}_{N} \psi \rangle &\geq& \mathcal{E}^{\text{HF}}_{N}(\omega_{N}) + (1 - \widetilde C\| \hat{v}\|_{1}) \langle R^{*}_{\omega_{N}}\psi, \mathbb{H}_{0} R^{*}_{\omega_{N}} \psi \rangle + \langle R^*_{\omega_{N}}\varphi, \mathbb{X} R^*_{\omega_{N}}\varphi \rangle\nonumber\\
&& + \langle R^{*}_{\omega_{N}} \psi, (\mathbb{F} + \mathbb{F}^{*}) R^{*}_{\omega_{N}} \psi \rangle - \widetilde C_{\delta} N^{-\frac{9}{16} + \delta}\nonumber\\
&\geq&  \mathcal{E}^{\text{HF}}_{N}(\omega_{N}) + (1 - C\|\hat{v}\|_{1}) \langle R^{*}_{\omega_{N}}\psi, \mathbb{H}_{0} R^{*}_{\omega_{N}} \psi \rangle \nonumber\\
&& + \langle R^{*}_{\omega_{N}} \psi, (\mathbb{F} + \mathbb{F}^{*}) R^{*}_{\omega_{N}} \psi \rangle - \widetilde C_{\delta} N^{-\frac{9}{16} + \delta}
\end{eqnarray}
where in the last step we used the bound (\ref{eq:Xbd}) for $\mathbb{X}$, and the improved bound for the number operator of Corollary \ref{cor:Nimpro}. We rewrite this last lower bound by isolating a small part of kinetic energy:
\begin{eqnarray}\label{eq:proofthm}
\langle \psi, \mathcal{H}_{N} \psi \rangle &\geq& \mathcal{E}^{\text{HF}}_{N}(\omega_{N}) + 2C\|\hat{v}\|_{1} \langle R^{*}_{\omega_{N}} \psi, \mathbb{H}_{0} R^{*}_{\omega_{N}} \psi \rangle  - \widetilde C_{\delta} N^{-\frac{9}{16} + \delta}\\
&& + (1-3C\|\hat{v}\|_{1}) \langle R^{*}_{\omega_{N}} \psi, \mathbb{H}_{0} R^{*}_{\omega_{N}} \psi \rangle + \langle R^{*}_{\omega_{N}} \psi, (\mathbb{F} + \mathbb{F}^{*}) R^{*}_{\omega_{N}} \psi \rangle\;,\nonumber
\end{eqnarray}
and we apply Lemma \ref{prp:lower2} to the last two terms in Eq. (\ref{eq:proofthm}), choosing $\alpha = 1 - 3 C\| \hat{v} \|_{1}$. We get:
\begin{eqnarray}
&&\langle \psi, \mathcal{H}_{N} \psi\rangle \geq  \mathcal{E}^{\text{HF}}_{N}(\omega_{N}) + 2 C\| \hat{v} \|_{1} \langle R^{*}_{\omega_{N}} \psi, \mathbb{H}_{0} R_{\omega_{N}}^{*} \psi \rangle - \widetilde C_{\delta} N^{-\frac{9}{16} + \delta} \nonumber\\&&  -\frac{1}{1 - 3C\| \hat{v} \|_{1}} \langle \Omega, \mathbb{F} \mathbb{H}_{0}^{-1} \mathbb{F}^{*} \Omega \rangle - \frac{C\| \hat{v} \|_{1}^{2}}{1-3C\| \hat{v} \|_{1}} \langle R^{*}_{\omega_{N}} \psi, \mathbb{H}_{0} R_{\omega_{N}}^{*} \psi \rangle\nonumber\\&& - C_{\delta} N^{-\frac{9}{16} + \delta}\;.
\end{eqnarray}
Our final claim, Eq. (\ref{eq:corrlow}), follows after suitably redefining the constant $C$.
\end{proof}
To conclude, let us discuss the proof of Lemma \ref{prp:lower2}.
\begin{proof} (of Lemma \ref{prp:lower2}.) The starting point is to rewrite, for any $\alpha > 0$:
\begin{eqnarray}\label{eq:chris}
\langle \varphi, \alpha \mathbb{H}_{0} \varphi \rangle  + \langle \varphi, (\mathbb{F} + \mathbb{F}^{*}) \varphi \rangle &=& \langle (1 + (\alpha \mathbb{H}_{0})^{-1} \mathbb{F}^{*})\varphi , \alpha \mathbb{H}_{0} (1 + (\alpha\mathbb{H}_{0})^{-1} \mathbb{F}^{*}) \varphi \rangle \nonumber\\
&& - \frac{1}{\alpha} \langle \varphi, \mathbb{F} \mathbb{H}_{0}^{-1} \mathbb{F}^{*} \varphi \rangle\nonumber\\
&\geq& - \frac{1}{\alpha} \langle \varphi, \mathbb{F} \mathbb{H}_{0}^{-1} \mathbb{F}^{*} \varphi \rangle\;,
\end{eqnarray}
where in the last step we used the positivity of $\mathbb{H}_{0}$. Now, the idea is to put the right-hand side of the inequality into normal order; the fully contracted term will give rise to the vacuum expectation in Eq. (\ref{eq:youngchris}), while the error terms will be proven to be either controlled by the kinetic energy or to be $O(N^{-\frac{9}{16} + \delta})$. We have, for any vector $\varphi \in \mathcal{F}$:
\begin{eqnarray}\label{eq:FstarF}
&&\mathbb{F} \mathbb{H}_{0}^{-1} \mathbb{F}^{*}\varphi \nonumber\\&&\qquad = \frac{1}{4 N^{2}} \sum_{p,q} \hat{v}(p) \hat{v}(q)\sum_{\substack{ k_{1},k_{2} \\ r_{1},r_{2} }}b_{k_{1}+p} c_{k_{1}} b_{k_{2}-p} c_{k_{2}} \mathbb{H}_{0}^{-1} b^{*}_{r_{1}+q} c^{*}_{r_{1}} b^{*}_{r_{2}-q} c^{*}_{r_{2}}\varphi\nonumber\\
&&\qquad =\frac{1}{4 N^{2}} \sum_{p,q} \hat{v}(p) \hat{v}(q)\sum_{\substack{ k_{1},k_{2} \\ r_{1},r_{2} }} b_{k_{1}+p} c_{k_{1}} b_{k_{2}-p} c_{k_{2}} b^{*}_{r_{1}+q} c^{*}_{r_{1}} b^{*}_{r_{2}-q} c^{*}_{r_{2}}\nonumber\\&&\qquad\quad\cdot \frac{1}{\mathbb{H}_{0} + E(r_{1},q) + E(r_{2},-q)} \varphi\;,
\end{eqnarray}
where $E(r,q) = e(r+q) + e(r)$. 
%The asterisk denotes the constraints:
%
%\begin{eqnarray}
%&&e(k_{1} + p) + e(k_{1}) \geq A|p|^{2}\eps^{2}\;,\qquad e(k_{2} - p) - e(k_{1}) \geq A|p|^{2}\eps^{2} \nonumber\\
%&&e(r_{1} + q) + e(r_{1}) \geq A|q|^{2}\eps^{2}\;,\qquad e(r_{2} - q) + e(r_{2}) \geq A|q|^{2}\eps^{2}\;,
%\end{eqnarray}
%
%recall the definition of $\mathbb{F}$, Eq. (\ref{eq:Fdef}). 
Eq. (\ref{eq:FstarF}) follows from the definition (\ref{eq:H-1}); observe that the operators $b^{*}_{r_{1} + q}$, $c^{*}_{r_{1}}$ add particles with energies $e(r_{1} + q)$, $e(r_{1})$, whose sum is $E(k_{1}, q)$, while the operator $b^{*}_{r_{2}-q}$, $c^{*}_{r_{2}}$ add particles with energies $e(r_{2}-q)$, $e(r_{2})$, whose sum is $E(k_{2},-q)$. Let us now put the fermionic monomial into normal order. We have:
\begin{eqnarray}
b_{k_{1}+p} c_{k_{1}} b_{k_{2}-p} c_{k_{2}} b^{*}_{r_{1}+q} c^{*}_{r_{1}} b^{*}_{r_{2}-q} c^{*}_{r_{2}} &=& \sum_{j=0}^{4} m_{2j}\;,
\end{eqnarray}
where $m_{2j}$ is a sum of normal-ordered fermionic monomials of order $2j$. Correspondingly, we rewrite
\begin{equation}\label{eq:Wj}
\mathbb{F} \mathbb{H}_{0}^{-1} \mathbb{F}^{*} = \langle \Omega, \mathbb{F} \mathbb{H}_{0}^{-1} \mathbb{F}^{*} \Omega \rangle + \sum_{j=1}^{4} \mathcal{W}_{2j}\;.
\end{equation}
We shall now estimate the $\mathcal{W}_{2j}$ terms, for $|p|\leq N^{\zeta}$, $|q|\leq N^{\zeta}$ for $\zeta$ as in Lemma \ref{lem:H0}. The contribution of the higher momenta can be proven to be small, thanks to the decay of the potential $\hat v(p)$. In fact, since $E(r_{1}, q)$ and $E(r_{2},-q)$ are greater than $\eps^{2}$, the contribution to (\ref{eq:FstarF}) of terms where at least one among $p$ and $q$ is greater that $N^{\zeta}$ is bounded in norm as:
\begin{equation}
\frac{C}{N^{2} \eps^{2}} (N\eps)^{4} N^{-\zeta K} \| |p|^{K} \hat v \|_{1} \|\varphi\| \leq C_{K} N^{\frac{4}{3} - \zeta K} \|\varphi\|\;,
\end{equation}
which is $o(\eps)$ for $K$ large enough.

Consider now the terms $\mathcal{W}_{2j}$ for momenta $|p|, |q|$ less than $N^{\zeta}$. Being normal ordered, the idea will be to control these terms with the kinetic energy $\mathbb{H}_{0}$ times a small constant, coming from the fact that all these terms are of second order in the interaction potential. To this end, let us introduce some notation. Let $p_{1}\equiv p$, $p_{2} \equiv -p$, $q_{1} \equiv q$, $q_{2}\equiv -q$. We say that two monomials $b_{k_{i}+p_{i}} c_{k_{i}}$ and $b^{*}_{r_{i} + q_{i}} c^{*}_{r_{i}}$ have been {\it contracted} if, in the normal ordering procedure giving rise to (\ref{eq:Wj}), they produced a nontrivial anticommutator. We shall denote by $a \in \{0,1,2\}$ the number of contracted $bc$, $b^{*} c^{*}$ monomials.

%\begin{remark} In what follows, the overall constants $C, C_{\delta}$ produced by the bounds might depend on $|p|$ and $|q|$ in a multiplicative fashion. The dependence will be at most quadratic in each momentum.
%\end{remark}

%We shall denote by $b^{*}_{r_{j'} + q_{j'}}$, $b_{k_{j'} + p_{j'}}$, $j'\in C^{b}$, the $b$ operators left after the contraction of the corresponding $c$-operators; and by $c^{*}_{r_{j''}}$, $c_{k_{j''}}$, $j''\in C^{c}$, the $c$ operators left after the contraction of the corresponding $b$-operators. Finally, we shall denote by $b^{*}_{r_{j} + q_{j}} c^{*}_{r_{j}}$, $j\in U_{1}$, and $b_{k_{i}+p_{i}} c_{k_{i}}$, $i\in U_{2}$, the {\it uncontracted} $b^{*}c^{*}$ and $bc$ monomials.

%The typical expression left after normal ordering is:
%
%\begin{eqnarray}\label{eq:gen}
%&&\langle \varphi, \prod_{j\in U_{1}} b^{*}_{r_{j} + q_{j}} c^{*}_{r_{j}} \prod_{j'\in C^{b}} b^{*}_{r_{j'} + q_{j'}} \prod_{j''\in C^{c}} c^{*}_{r_{j''}} \\\
%&&\quad \cdot \prod_{i\in U_{2}} b_{k_{i} + p_{i}} c_{k_{i}} \prod_{j'\in C^{b}} b_{k_{j'} + p_{j'}} \prod_{j''\in C^{c}} c_{k_{j''}} \frac{1}{\mathbb{H}_{0} + E(\underline{r}, \underline{q})} \varphi\rangle\cdot \nonumber\\&&\quad \prod_{j'\in C^{b}} \delta_{r_{j'} + q_{j'}, k_{j'} + p_{j'}} \prod_{j''\in C^{c}} \delta_{r_{j''}, k_{j''}}\;,\nonumber
%\end{eqnarray}
%
%With $E(\underline{r}, \underline{q}) = E(r_{1}, q_{1}) + E(r_{2}, q_{2})$. Let $a = |U_{1}| = |U_{2}|$. We shall estimate the few inequivalent cases separately.
\medskip

\noindent{\underline{Case $a=0$.}} Suppose that no fermionic operators have been contracted. Let $E(\underline{r}, \underline{q}) = E(r_{1}, q_{1}) + E(r_{2}, q_{2})$. The corresponding contribution to (\ref{eq:Wj}) is:
\begin{eqnarray}\label{eq:gen0}
&&\Big| \langle \varphi, b^{*}_{r_{1}+q_{1}} c^{*}_{r_{1}} b^{*}_{r_{2}+q_{2}} c^{*}_{r_{2}}  b_{k_{1}+p_{1}} c_{k_{1}} b_{k_{2}+p_{2}} c_{k_{2}} \frac{1}{\mathbb{H}_{0} + E(\underline{r}, \underline{q})} \varphi \rangle\Big| \nonumber\\
&&\leq \Big\|  c_{r_{2}} b_{r_{2}+q_2} c_{r_{1}} b_{r_{1}+q_1} \frac{1}{(\mathbb{H}_{0} + E(\underline{k}, \underline{p}))^{1/2}} \varphi  \Big\|\nonumber\\
&&\quad \cdot \Big\| b_{k_{1}+p_1} c_{k_{1}} b_{k_{2}+p_2} c_{k_{2}} \frac{1}{(\mathbb{H}_{0} + E(\underline{r}, \underline{q}))^{1/2}} \varphi \Big\|\;.
\end{eqnarray}
Therefore, by Lemma \ref{lem:H0}:
\begin{eqnarray}\label{eq:a=2}
&&\sum_{\substack{k_{i}, r_{j}: \\ k_{i}, r_{j} \in \mathcal{B}_{\mu} \\ r_{i} + q_{i}, k_{j} + p_{j} \in \mathcal{B}_{\mu}^{c}}}|(\ref{eq:gen0})| \nonumber\\
&&\leq CN \sum^{*}_{r_{1}, k_{2}}  \Big\|  \mathbb{H}_{0}^{1/2} c_{r_{1}} b_{r_{1}+q_{1}} \frac{1}{(\mathbb{H}_{0} + E(k_{2}, p_{2}))^{1/2}} \varphi  \Big\| \nonumber\\
&&\quad\qquad\qquad \cdot \Big\| \mathbb{H}_{0}^{1/2} b_{k_{2}+p_{2}} c_{k_{2}} \frac{1}{(\mathbb{H}_{0} + E(r_{1}, q_{1}))^{1/2}} \varphi \Big\| \nonumber\\
&&\leq CN \sum^{*}_{r_{1}, k_{2}}  \Big\|   c_{r_{1}} b_{r_{1}+q_{1}} \varphi  \Big\| \Big\| b_{k_{2}+p_{2}} c_{k_{2}} \varphi \Big\| \nonumber\\
&&\leq C N^{2} \langle \varphi, \mathbb{H}_{0} \varphi \rangle\;,
\end{eqnarray}
where the asterisk recalls the constraints on the momenta. The first inequality in Eq. (\ref{eq:a=2}) follows from Lemma \ref{lem:H0} and from the trivial bounds $E(\underline{r}, \underline{q}) \geq E(r_{1}, q_{1})$, $E(\underline{k}, \underline{p}) \geq E(k_{2}, p_{2})$. The second inequality follows from:
\begin{eqnarray}
&&\mathbb{H}_{0}^{1/2} c_{r_{1}} b_{r_{1}+q_{1}} \frac{1}{(\mathbb{H}_{0} + E(k_{2}, p_{2}))^{1/2}} \\&&\qquad\qquad\qquad =   \frac{\mathbb{H}_{0}^{1/2}}{(\mathbb{H}_{0} + E(r_{1},q_{1}) + E(k_{2}, p_{2}))^{1/2}} c_{r_{1}} b_{r_{1}+q_{1}}\nonumber\\
&&\mathbb{H}_{0}^{1/2} b_{k_{2} + p_{2}} c_{k_{2}} \frac{1}{(\mathbb{H}_{0} + E(r_{1}, q_{1}))^{1/2}}\nonumber\\ &&\qquad\qquad\qquad =   \frac{\mathbb{H}_{0}^{1/2}}{(\mathbb{H}_{0} + E(k_{2},p_{2}) + E(r_{1}, q_{1}))^{1/2}} b_{k_{2} + p_{2}} c_{k_{2}}\nonumber
\end{eqnarray}
and from
\begin{equation}
\frac{\mathbb{H}_{0}^{1/2}}{(\mathbb{H}_{0} + E(k_{2},p_{2}) + E(r_{1}, q_{1}))^{1/2}} \leq 1\;.
\end{equation}
The last step in Eq. (\ref{eq:a=2}) follows again from Lemma \ref{lem:H0}. This concludes the discussion of the $a=0$ case.
\medskip

\noindent{{\underline{Case $a=1$.}}} Suppose that only one pair of $c$, $c^{*}$ operators has been contracted. The corresponding contribution is:
\begin{eqnarray}\label{eq:gen2}
&&\Big|\langle \varphi, b^{*}_{r_{1} + q_{1}} c^{*}_{r_{1}} b^{*}_{r_{2} + q_{2}} b_{k_{1} + p_{1}} c_{k_{1}} b_{r_{2} + p_{2}} \frac{1}{\mathbb{H}_{0} + E(\underline{r}, \underline{q})} \varphi\rangle\Big|\delta_{k_{2},r_{2}}\\
&& \leq \Big\|  b_{r_{2} + q_{2}} c_{r_{1}} b_{r_{1} + q_{1}}  \frac{1}{(\mathbb{H}_{0} + E(k_{1}, p_{1}) + E(r_{2}, p_{2}))^{1/2}} \varphi \Big\|\nonumber\\&&\quad \cdot \Big\| b_{k_{1} + p_{1}} c_{k_{1}} b_{r_{2} + p_{2}} \frac{1}{(\mathbb{H}_{0} + E(\underline{r}, \underline{q}))^{1/2}} \varphi \Big\| \delta_{k_{2}, r_{2}}\;.\nonumber
\end{eqnarray}
Therefore, by Lemma \ref{lem:H0}:
\begin{eqnarray}\label{eq:case1}
&&\sum_{\substack{k_{i}, r_{j}: \\ k_{i}, r_{j} \in \mathcal{B}_{\mu} \\ r_{i} + q_{i}, k_{j} + p_{j} \in \mathcal{B}_{\mu}^{c}}}|(\ref{eq:gen2})| \\
&&\leq CN\cdot  \sum_{\substack{r_{2}:\, r_{2} \in \mathcal{B}_{\mu}\\ r_{2} + p_{2} \in \mathcal{B}_{\mu}^{c},\, r_{2} + q_{2} \in \mathcal{B}_{\mu}^{c}}} \Big\| \mathbb{H}_{0}^{1/2}b_{r_{2} + q_{2}} \frac{1}{(\mathbb{H}_{0} + E(r_{2},p_{2}))^{1/2}} \varphi \Big\| \nonumber\\&&\qquad\qquad\qquad\qquad\qquad\cdot \Big\| \mathbb{H}_{0}^{1/2} b_{r_{2} + p_{2}}  \frac{1}{(\mathbb{H}_{0} + E(r_{2}, q_{2}))^{1/2}}\varphi \Big\|\nonumber\\
&&\leq CN\cdot  \sum_{\substack{r_{2}:\, r_{2} \in \mathcal{B}_{\mu}\\ r_{2} + p_{2} \in \mathcal{B}_{\mu}^{c},\, r_{2} + q_{2} \in \mathcal{B}_{\mu}^{c}}} \Big\| b_{r_{2} + q_{2}} \frac{\mathbb{H}_{0}^{1/2}}{(\mathbb{H}_{0} + E(r_{2},p_{2}))^{1/2}} \varphi \Big\| \nonumber\\&&\qquad\qquad\qquad\qquad\qquad\cdot \Big\|  b_{r_{2} + p_{2}}  \frac{\mathbb{H}_{0}^{1/2}}{(\mathbb{H}_{0} + E(r_{2}, q_{2}))^{1/2}}\varphi \Big\|\;.\nonumber
\end{eqnarray}
We now split the sum as:
\begin{eqnarray}
\sum_{\substack{r_{2}:\, r_{2} \in \mathcal{B}_{\mu}\\ r_{2} + p_{2} \in \mathcal{B}_{\mu}^{c},\, r_{2} + q_{2} \in \mathcal{B}_{\mu}^{c}}} &=& \sum^{*}_{\substack{r_{2}: \\ \text{$e(r_{2} + q_{2}) \leq 1/N$ or $e(r_{2} + p_{2}) \leq 1/N$}}}\nonumber\\&& + \sum^{*}_{\substack{r_{2}: \\ \text{$e(r_{2} + q_{2}) > 1/N$ and $e(r_{2} + p_{2}) > 1/N$}}}
\end{eqnarray}
where the asterisk recalls the original constraints $r_{2}\in \mathcal{B}_{\mu}$, $r_{2} + p_{2} \in \mathcal{B}_{\mu}^{c}$, $r_{2} + q_{2} \in \mathcal{B}_{\mu}^{c}$. Correspondingly, we rewrite the right-hand side of Eq. (\ref{eq:case1}) as $\text{I} + \text{II}$. Consider $\text{I}$. Using that the number of terms in the sum is bounded by $C_{\delta} N^{1 - \frac{9}{16} + \delta}$ (recall Eq. (\ref{eq:numbth})), that $\|b_{k}\|\leq 1$ and that $\|\varphi\|=1$:
\begin{equation}\label{eq:gen3I}
|\text{I}| \leq C_{\delta}N^{2 - \frac{9}{16} + \delta}\;.
\end{equation}
Consider now $\text{II}$. By the Cauchy-Schwarz inequality, we have:
\begin{eqnarray}\label{eq:gen3II}
|\text{II}| &\leq& CN^{2} \Big(\sum^{*}_{r_{2}:\, e(r_{2}+q_{2})>1/N } \frac{1}{N} \Big\| b_{r_{2} + q_{2}} \frac{\mathbb{H}_{0}^{1/2}}{(\mathbb{H}_{0} + E(r_{2},p_{2}))^{1/2}} \varphi \Big\|^{2}\Big)^{1/2} \nonumber\\&&\cdot \Big( \sum^{*}_{r_{2}:\, e(r_{2}+p_{2})>1/N } \frac{1}{N} \Big\|  b_{r_{2} + p_{2}}  \frac{\mathbb{H}_{0}^{1/2}}{(\mathbb{H}_{0} + E(r_{2}, q_{2}))^{1/2}}\varphi \Big\|  \Big)^{1/2}\nonumber\\
&\leq& CN^{2} \langle \varphi, \mathbb{H}_{0} \varphi \rangle\;,
\end{eqnarray}
where in the last step we used that:
\begin{equation}
b_{r_{2} + q_{2}} \frac{\mathbb{H}_{0}^{1/2}}{(\mathbb{H}_{0} + E(r_{2},p_{2}))^{1/2}} = \frac{(\mathbb{H}_{0} + e(r_{2} + q_{2}))^{1/2}}{(\mathbb{H}_{0} + e(r_{2} + q_{2})+ E(r_{2},p_{2}))^{1/2}} b_{r_{2} + q_{2}} 
\end{equation}
and that:
\begin{equation}
\frac{(\mathbb{H}_{0} + e(r_{2} + q_{2}))^{1/2}}{(\mathbb{H}_{0} + e(r_{2} + q_{2})+ E(r_{2},p_{2}))^{1/2}} \leq 1\;.
\end{equation}
Therefore,
\begin{equation}
\sum_{\substack{k_{i}, r_{j}: \\ k_{i}, r_{j} \in \mathcal{B}_{\mu} \\ r_{i} + q_{i}, k_{j} + p_{j} \in \mathcal{B}_{\mu}^{c}}}|(\ref{eq:gen2})| \leq C_{\delta}N^{2 - \frac{9}{16} + \delta} + CN^{2} \langle \varphi, \mathbb{H}_{0} \varphi \rangle\;.
\end{equation}
Suppose now that only one pair of $b, b^{*}$ operators has been contracted. We have:
\begin{eqnarray}\label{eq:gen3}
&&\Big|\langle \varphi, b^{*}_{r_{1} + q_{1}} c^{*}_{r_{1}} c^{*}_{r_{2}} b_{k_{1} + p_{1}} c_{k_{1}} c_{r_{2} - p_{2} + q_{2}} \frac{1}{\mathbb{H}_{0} + E(\underline{r}, \underline{q})} \varphi\rangle\Big|\delta_{k_{2}+p_{2}, r_{2}+q_{2}} \nonumber\\
&&\leq  \Big\| c_{r_{2}} c_{r_{1}} b_{r_{1}+q_{1}} \frac{1}{(\mathbb{H}_{0} + E(r_{2} + q_{2}, -p_{2}) + E(k_{1}, p_{1}))^{1/2}}\varphi \Big\| \\
&&\quad\cdot \Big\| b_{k_{1} + p_{1}} c_{k_{1}} c_{r_{2} - p_{2} + q_{2}} \frac{1}{(\mathbb{H}_{0} + E(r_{1}, q_{1}) + E(r_{2}, q_{2}))^{1/2}} \varphi \Big\|\delta_{k_{2}+p_{2}, r_{2}+q_{2}}\;.\nonumber
\end{eqnarray}
Therefore, by Lemma \ref{lem:H0}:
\begin{eqnarray}\label{eq:gen4}
&&\sum_{\substack{k_{i}, r_{j}: \\ k_{i}, r_{j} \in \mathcal{B}_{\mu} \\ r_{i} + q_{i}, k_{j} + p_{j} \in \mathcal{B}_{\mu}^{c}}}|(\ref{eq:gen3})| \\
&&\leq CN\cdot  \sum_{\substack{r_{2}: \\ r_{2}\in \mathcal{B}_{\mu}, r_{2} + p_{2} \in \mathcal{B}_{\mu}^{c} \\ r_{2} + q_{2} - p_{2} \in \mathcal{B}_{\mu},\, r_{2} + q_{2} \in \mathcal{B}_{\mu}^{c}}}\Big\| \mathbb{H}_{0}^{1/2} c_{r_{2}} \frac{1}{(\mathbb{H}_{0} + E(r_{2} + q_{2}, -p_{2}))^{1/2}}\varphi \Big\| \nonumber\\
&&\qquad\qquad\qquad\qquad\qquad\cdot \Big\| \mathbb{H}_{0}^{1/2} c_{r_{2} - p_{2} + q_{2}} \frac{1}{(\mathbb{H}_{0} + E(r_{2}, q_{2}))^{1/2}} \varphi \Big\|\;.\nonumber
\end{eqnarray}
This term can be bounded proceeding as for the previous one. One has $(\ref{eq:gen4}) = \text{I} + \text{II}$, with $\text{I}$, $\text{II}$ bounded as in Eqs. (\ref{eq:gen3I}), (\ref{eq:gen3II}), respectively; we omit the details. The final result is:
\begin{equation}
\sum_{\substack{k_{i}, r_{j}: \\ k_{i}, r_{j} \in \mathcal{B}_{\mu} \\ r_{i} + q_{i}, k_{j} + p_{j} \in \mathcal{B}_{\mu}^{c}}}|(\ref{eq:gen3})| \leq C_{\delta}N^{2 - \frac{9}{16} + \delta}  + CN^{2} \langle \varphi, \mathbb{H}_{0} \varphi \rangle\;.
\end{equation}
To conclude the discussion of the case $a=1$, suppose that one $bc$ monomial and one $b^{*}c^{*}$ monomial have been fully contracted. The corresponding contribution to (\ref{eq:Wj}) is:
\begin{eqnarray}\label{eq:gen6}
&&\Big|\langle \varphi, b^{*}_{r_{1} + q_{1}} c^{*}_{r_{1}} b_{k_{1} + p_{1}} c_{k_{1}} \frac{1}{\mathbb{H}_{0} + E(\underline{r}, \underline{q})} \varphi\rangle\Big|\delta_{k_{2}, r_{2}}\delta_{p_{2}, q_{2}} \nonumber\\
&&\leq \Big\| c_{r_{1}} b_{r_{1} + q_{1}} \frac{1}{(\mathbb{H}_{0} + E(r_{2}, q_{2}) + E(k_{1}, p_{1}))^{1/2}} \varphi \Big\| \nonumber\\
&&\quad \cdot \Big\| b_{k_{1} + p_{1}} c_{k_{1}} \frac{1}{(\mathbb{H}_{0} + E(r_{1}, q_{1}) + E(r_{2}, q_{2}))^{1/2}}\varphi \Big\| \delta_{k_{2},r_{2}}\delta_{p_{2},q_{2}} \;.
\end{eqnarray}
Therefore, by Lemma \ref{lem:H0}:
\begin{eqnarray}\label{eq:a1}
&&\sum_{\substack{k_{i}, r_{j}: \\ k_{i}, r_{j} \in \mathcal{B}_{\mu} \\ r_{i} + q_{i}, k_{j} + p_{j} \in \mathcal{B}_{\mu}^{c}}} |(\ref{eq:gen6})| \nonumber\\
&&\leq  CN\cdot  \sum_{\substack{r_{2}: \\ r_{2} \in \mathcal{B}_{\mu}, r_{2} + q_{2} \in \mathcal{B}_{\mu}^{c}}} \Big\| \frac{\mathbb{H}_{0}^{1/2}}{( \mathbb{H}_{0} + E(r_{2}, q_{2}) )^{1/2}} \varphi\Big\|^{2}\delta_{p_{2},q_{2}}\nonumber\\
&&\leq CN \langle \varphi, \mathbb{H}_{0} \varphi \rangle \sum_{\substack{r_{2}: \\ r_{2} \in \mathcal{B}_{\mu}, r_{2} + q_{2} \in \mathcal{B}_{\mu}^{c}}} \frac{1}{E(r_{2}, q_{2})}\delta_{p_{2},q_{2}}\nonumber\\
&&\leq CN^{2}  \langle \varphi, \mathbb{H}_{0} \varphi \rangle\delta_{p_{2},q_{2}}\;.
\end{eqnarray}
This concludes the discussion of the $a=1$ case.
\medskip

\noindent{\underline{Case $a=2$.}} Consider now the case in which all $bc$, $b^{*}c^{*}$ monomials have been contracted. Suppose only the $c$, $c^{*}$ operators have been contracted. We have:
\begin{eqnarray}\label{eq:gen7}
&&\Big|\langle \varphi, b^{*}_{k_{1} + q_{1}} b^{*}_{k_{2}+q_{2}} b_{k_{1} + p_{1}} b_{k_{2} + p_{2}} \frac{1}{\mathbb{H}_{0} + E(\underline{k}, \underline{q})} \varphi \rangle\Big|\delta_{k_{1},r_{1}}\delta_{k_{2}, r_{2}} \nonumber\\
&&\leq \Big\| b_{k_{2}+q_{2}} b_{k_{1}+q_{1}} \frac{1}{(\mathbb{H}_{0} + E(k_{1}, p_{1}) + E(k_{2}, p_{2}))^{1/2}} \varphi \Big\|\nonumber\\
&&\quad \cdot  \Big\| b_{k_{1}+p_{1}} b_{k_{2}+p_{2}} \frac{1}{( \mathbb{H}_{0} + E(k_{1},q_{1}) + E(k_{2}, q_{2}) )^{1/2}} \Big\| \delta_{k_{1},r_{1}}\delta_{k_{2},r_{2}} \;.
\end{eqnarray}
Thus, by Cauchy-Schwarz:
\begin{eqnarray}\label{eq:gen8}
&&\sum_{\substack{k_{i}, r_{j}: \\ k_{i}, r_{j} \in \mathcal{B}_{\mu} \\ r_{i} + q_{i}, k_{j} + p_{j} \in \mathcal{B}_{\mu}^{c}}} |(\ref{eq:gen7})| \\
&&\leq \sum_{\substack{k_{1}, k_{2}: \\ k_{1}, k_{2} \in \mathcal{B}_{\mu}, k_{i}+p_{i} \in \mathcal{B}_{\mu}^{c}}}\Big\| b_{k_{2}+q_{2}} b_{k_{1}+q_{1}} \frac{1}{(\mathbb{H}_{0} + E(k_{1}, p_{1}) + E(k_{2}, p_{2}))^{1/2}} \varphi \Big\|^{2} \nonumber\\
&& \quad + (p_{i}\leftrightarrow q_{i})\;.
\end{eqnarray}
Consider the first term in the right-hand side. We have the following possibilities: only one among $k_{1}$ and $k_{2}$ is such that $e(k_{i} + q_{i}) < 1/N$; both $k_{1}$ and $k_{2}$ are such that $e(k_{i} + q_{i}) < 1/N$; both $k_{1}$ and $k_{2}$ are such that $e(k_{i} + q_{i}) \geq 1/N$. Correspondingly, we shall rewrite $(\ref{eq:gen8}) = \text{I} +\text{II} + \text{III}$. Consider $\text{I}$. We have:
\begin{eqnarray}\label{eq:estI}
|\text{I}| &\leq& CN\cdot  \sum_{\substack{k_{1}, k_{2}: \\ k_{1} \in \mathcal{B}_{\mu}, k_{1} + p_{1}\in \mathcal{B}_{\mu}^{c} \\ e(k_{1} + q_{1}) > 1/N, e(k_{2} + q_{2})\leq 1/N}} \frac{1}{N}\Big\| b_{k_{1}+q_{1}} \frac{1}{(\mathbb{H}_{0} + E(k_{1}, p_{1}))^{1/2}} \varphi \Big\|^{2} \nonumber\\
&\leq& CN\cdot  \sum_{k_{2}:\, e(k_{2} + q_{2}) \leq 1/N} \sum_{k_{1}} e(k_{1} + q_{1})\Big\| b_{k_{1}+q_{1}} \frac{1}{\mathbb{H}_{0}^{1/2}} \varphi \Big\|^{2} \nonumber\\
&\leq& C_{\delta}N^{2 - \frac{9}{16} + \delta}
\end{eqnarray}
where the last step follows from Eq. (\ref{eq:numbth}).  Consider $\text{II}$. Using again that the number of momenta $k_{i} \in 2\pi \mathbb{Z}^{3}$ satisfying the constraint $e(k_{i} + q_{i}) < 1/N$ is bounded by $C_{\delta} N^{1 - \frac{9}{16} + \delta}$ (recall Eq. (\ref{eq:numbth})), we have, by Lemma \ref{lem:H0}:
\begin{equation}\label{eq:estII}
|\text{II}| \leq C_{\delta}N^{1 - \frac{9}{16} + \delta} \sum_{\substack{k_{1}: \\ k_{1}\in \mathcal{B}_{\mu}, k_{1}+p_{1} \in \mathcal{B}_{\mu}^{c}}} \frac{1}{E(k_{1}, p_{1})} \leq C_{\delta}N^{2 - \frac{9}{16} + \delta}\;.
\end{equation}
Consider $\text{III}$. We have:
\begin{eqnarray}\label{eq:gen10b}
|\text{III}| &\leq& \sum_{\substack{k_{1}, k_{2}: \\ k_{1}, k_{2} \in \mathcal{B}_{\mu}, k_{i}+p_{i} \in \mathcal{B}_{\mu}^{c} \\ e(k_{i} + q_{i}) > 1/N}}\Big\| b_{k_{2}+q_{2}} b_{k_{1}+q_{1}} \frac{1}{(\mathbb{H}_{0} + E(k_{1}, p_{1}) + E(k_{2}, p_{2}))^{1/2}} \varphi \Big\|^{2} \nonumber\\
&\leq& C  N^{2} \langle \varphi, \mathbb{H}_{0} \varphi \rangle\;.
\end{eqnarray}
Therefore, combining Eqs. (\ref{eq:estI}), (\ref{eq:estII}), (\ref{eq:gen10b}) we get:
\begin{equation}
\sum_{\substack{k_{i}, r_{j}: \\ k_{i}, r_{j} \in \mathcal{B}_{\mu} \\ r_{i} + q_{i}, k_{j} + p_{j} \in \mathcal{B}_{\mu}^{c}}} |(\ref{eq:gen7})| \leq C_{\delta}N^{2 - \frac{9}{16} + \delta} + C N^{2} \langle \varphi, \mathbb{H}_{0} \varphi \rangle\;.
\end{equation}
Similar bounds are obtained in case the contractions involve the $b, b^{*}$ operators. Consider now the case in which one pair of $c$, $c^{*}$ operators and one pair of $b$, $b^{*}$ operators have been contracted. We have:
\begin{eqnarray}\label{eq:gen10}
&&\Big|\langle \varphi, b^{*}_{r_{1}+q_{1}} c^{*}_{r_{2}} b_{k_{1}+p_{1}} c_{k_{2}} \frac{1}{\mathbb{H}_{0} + E(\underline{r},\underline{q})} \varphi \rangle\Big|\delta_{k_{1},r_{1}}\delta_{k_{2}+p_{2}, r_{2}+q_{2}}\nonumber\\
&&\leq \Big\| c_{r_{2}} b_{r_{1}+q_{1}}\frac{1}{(\mathbb{H}_{0} + E(r_{1}, p_{1}) + E(k_{2}, p_{2}))^{1/2}}  \varphi\Big\| \\
&&\quad \cdot \Big\|b_{k_{1}+p_{1}} c_{k_{2}} \frac{1}{(\mathbb{H}_{0} + E(r_{1},q_{1}) + E(r_{2}, q_{2}))^{1/2}} \varphi \Big\| \delta_{k_{1},r_{1}}\delta_{k_{2}+p_{2}, r_{2}+q_{2}}\;.\nonumber
\end{eqnarray}
Therefore, by Cauchy-Schwarz:
\begin{eqnarray}
&&\sum_{\substack{k_{i}, r_{j}: \\ k_{i}, r_{j} \in \mathcal{B}_{\mu} \\ r_{i} + q_{i}, k_{j} + p_{j} \in \mathcal{B}_{\mu}^{c}}}|(\ref{eq:gen10})| \leq \nonumber\\
&& \leq \sum_{\substack{k_{1}, k_{2}: \\ k_{i}\in \mathcal{B}_{\mu},\, k_{i}+ p_{i} \in \mathcal{B}_{\mu}^{c}}} \Big\| c_{k_{2} + p_{2} - q_{2}} b_{k_{1}+q_{1}}\frac{1}{(\mathbb{H}_{0} + E(k_{1}, p_{1}) + E(k_{2}, p_{2}))^{1/2}}  \varphi\Big\|^{2} +\nonumber\\
&& \sum_{\substack{k_{1}, k_{2}: \\ k_{i}\in \mathcal{B}_{\mu},\, k_{i}+ p_{i} \in \mathcal{B}_{\mu}^{c} \\ k_{2} + p_{2} - q_{2} \in \mathcal{B}_{\mu},\, k_{1} + q_{1} \in \mathcal{B}_{\mu}^{c}}} \Big\|b_{k_{1}+p_{1}} c_{k_{2}} \frac{1}{(\mathbb{H}_{0} + E(k_{1},q_{1}) + E(k_{2} + p_{2} - q_{2}, q_{2}))^{1/2}} \varphi \Big\|^{2}\;. \nonumber
\end{eqnarray}
Both terms can be studied as in (\ref{eq:gen8})-(\ref{eq:gen10b}), we omit the details; the result is:
\begin{equation}
\sum_{\substack{k_{i}, r_{j}: \\ k_{i}, r_{j} \in \mathcal{B}_{\mu} \\ r_{i} + q_{i}, k_{j} + p_{j} \in \mathcal{B}_{\mu}^{c}}}|(\ref{eq:gen10})| \leq C_{\delta}N^{2 - \frac{9}{16} + \delta} + C N^{2} \langle \varphi, \mathbb{H}_{0} \varphi \rangle\;.
\end{equation}
Finally, suppose that all the $c, c^{*}$ operators and one pair of $b, b^{*}$ operators have been contracted. The corresponding contribution is:
\begin{eqnarray}\label{eq:gen9}
&&\Big| \langle \varphi, b^{*}_{k_{1}+q_{1}} b_{k_{2} + p_{2}} \frac{1}{\mathbb{H}_{0} + E(\underline{k}, \underline{q})} \varphi \rangle \Big|\delta_{r_{1}, k_{1}}\delta_{r_{2}, k_{2}}\delta_{k_{2} + p_{2}, k_{1} + p_{1}}   \nonumber\\
&&\leq \Big\| b_{k_{1} + q_{1}} \frac{1}{(\mathbb{H}_{0} + E(k_{2}, q_{2}) + E(k_{1}, p_{1}))^{1/2}}\varphi \Big\| \\&&\quad \cdot \Big\| b_{k_{2} + p_{2}} \frac{1}{(\mathbb{H}_{0} + E(k_{1}, q_{1}) + E(k_{2}, q_{2}))^{2}}\Big\|\delta_{r_{1}, k_{1}} \delta_{r_{2}, k_{2}}\delta_{k_{2} + p_{2}, k_{1} + p_{1}} \;.\nonumber
\end{eqnarray}
Then,
\begin{eqnarray}\label{eq:a2}
&&\sum_{\substack{k_{i}, r_{j}: \\ k_{i}, r_{j} \in \mathcal{B}_{\mu} \\ r_{i} + q_{i}, k_{j} + p_{j} \in \mathcal{B}_{\mu}^{c}}}|(\ref{eq:gen9})| \leq \nonumber\\
&& \sum_{\substack{ k_{1}:\, k_{1} \in \mathcal{B}_{\mu},\, k_{1} + p_{1} \in \mathcal{B}_{\mu}^{c} \\ k_{1} + q_{1} \in \mathcal{B}^{c}_{\mu}}} \Big\| b_{k_{1} + q_{1}} \frac{1}{(\mathbb{H}_{0} + E(k_{1}+p_{1} - p_{2}, q_{2}) + E(k_{1}, p_{1}))^{1/2}}\varphi \Big\| \nonumber\\&&\qquad\qquad\qquad\qquad  \cdot \Big\| b_{k_{1} + p_{1}} \frac{1}{(\mathbb{H}_{0} + E(k_{1}, q_{1}) + E(k_{1} + p_{1} - p_{2}, q_{2}))^{2}} \varphi\Big\|\nonumber\\
&&\qquad\qquad\leq \sum_{\substack{ k_{1}:\, k_{1} \in \mathcal{B}_{\mu},\, k_{1} + p_{1}\in \mathcal{B}_{\mu}^{c} }} \frac{1}{E(k_{1}, p_{1})} + \sum_{\substack{ k_{1}:\, k_{1} \in \mathcal{B}_{\mu},\, k_{1} + q_{1}\in \mathcal{B}_{\mu}^{c} }} \frac{1}{E(k_{1}, q_{1})}  \nonumber\\&&\qquad\qquad \leq CN\;,
\end{eqnarray} 
where the last step follows from Lemma \ref{lem:H0}. The remaining cases can be studied in exactly the same way. This concludes the discussion of the $a=2$ case.

\medskip

\noindent{\underline{Conclusion.}} All in all the bounds (\ref{eq:a=2}), (\ref{eq:a1}), (\ref{eq:a2}) imply, for $|p|, |q|\leq N^{\zeta}$, recalling Eqs. (\ref{eq:FstarF}), (\ref{eq:Wj}):
\begin{equation}\label{eq:Wjest}
\langle \varphi, \mathcal{W}_{2j}\varphi \rangle \leq C_{\delta}N^{-\frac{9}{16} + \delta} + C\| \hat{v} \|^{2}_{1} \langle \varphi, \mathbb{H}_{0} \varphi \rangle\;.
\end{equation}
Therefore, by Eqs. (\ref{eq:chris}), (\ref{eq:Wj}), (\ref{eq:Wjest}):
\begin{eqnarray}
\langle \varphi, \alpha \mathbb{H}_{0} \varphi \rangle  + \langle \varphi, (\mathbb{F} + \mathbb{F}^{*}) \varphi \rangle &\geq& - \frac{1}{\alpha} \langle \varphi, \mathbb{F} \mathbb{H}_{0}^{-1} \mathbb{F}^{*} \varphi \rangle\\
&\geq& - \frac{1}{\alpha} \langle \Omega, \mathbb{F} \mathbb{H}_{0}^{-1} \mathbb{F}^{*} \Omega \rangle - C\| \hat{v} \|_{1}^{2} \langle \varphi, \mathbb{H}_{0} \varphi \rangle\nonumber\\&& - C_{\delta}N^{- \frac{9}{16} + \delta} \;.\nonumber
\end{eqnarray}
This concludes the proof of Lemma \ref{prp:lower2}.
\end{proof}

\subsection{Proof of Theorem \ref{thm:1}: upper bound}\label{sec:up}

To conclude the proof of Theorem \ref{thm:1}, in this section we shall compute an upper bound for the correlation energy.
\begin{proposition}[Upper bound for the correlation energy]\label{prp:upper} Under the same assumptions of Theorem \ref{thm:cor}, the following is true:
\begin{equation}\label{eq:upcor}
E_{N} \leq \mathcal{E}^{\text{HF}}_{N}(\omega_{N}) - (1 - C\|\hat{v}\|_{1}) \langle \Omega, \mathbb{F} \mathbb{H}_{0}^{-1} \mathbb{F}^{*}\Omega \rangle + C_{\delta}N^{- \frac{9}{16} + \delta}\;.
\end{equation}
\end{proposition}
Eq. (\ref{eq:upcor}) proves the upper bound on the correlation energy in Theorem \ref{thm:1}, and concludes the proof of our main result. The proof of Proposition \ref{prp:upper} is based on the choice of a suitable trial state. We define:
\begin{equation}\label{eq:trial}
\psi_{\mathbb{F}} = M_{\mathbb{F}}^{-1} R_{\omega_{N}} (1 - \mathbb{H}_{0}^{-1} \mathbb{F}^{*}) \Omega\;,
\end{equation}
where the normalization factor $M_{\mathbb{F}}$ is chosen to ensure that $\| \psi_{\mathbb{F}} \| = 1$, that is $M_{\mathbb{F}}^{2} = 1 + \langle \Omega, \mathbb{F} \mathbb{H}_{0}^{-2} \mathbb{F}^{*} \Omega \rangle$. As proven in Appendix \ref{app:norm}, we have:
\begin{equation}\label{eq:normbd}
|M_{\mathbb{F}} - 1|\leq C\| \hat{v} \|^{2}_{1} + \frak{e}_{K}\;,
\end{equation}
with $\frak{e}_{K} = o(1)$, recall Eqs. (\ref{eq:Qs}). It is not difficult to see that $\psi_{\mathbb{F}}$ is an $N$-particle state, $\psi_{\mathbb{F}} \in \mathcal{F}^{(N)}$. Since $R_{\omega_{N}} \Omega$ is an $N$-particle state, it is enough to show that $R_{\omega_{N}}\mathbb{H}_{0}^{-1} \mathbb{F}^{*} \Omega$ is an $N$-particle state. Proceeding as in Eq. (\ref{eq:FstarF}):
\begin{eqnarray}\label{eq:Npart}
&&R_{\omega_{N}} \mathbb{H}_{0}^{-1} \mathbb{F}^{*} \Omega = R_{\omega_{N}} \frac{1}{2N} \sum_{p, k, k'} \hat v(p) b^{*}_{k+q} c^{*}_{k} b^{*}_{k'-q} c^{*}_{k'} \frac{1}{E(k, q) + E(k',-q)} \Omega\nonumber\\
&&\quad = -\frac{1}{2N} \sum_{p, k, k'} \hat v(p) b^{*}_{k+q} b^{*}_{k'-q} c_{k}  c_{k'}  \frac{1}{E(k, q) + E(k',-q)} R_{\omega_{N}}\Omega\;,
\end{eqnarray}
where in the last step we used that $R^{*}_{\omega_{N}} b_{k} R_{\omega_{N}} = b_{k}$, $R^{*}_{\omega_{N}} c_{k} R_{\omega_{N}} = c^{*}_{k}$, recall Eq. (\ref{eq:propbog}). Thus, $R_{\omega_{N}} \mathbb{H}_{0}^{-1} \mathbb{F}^{*} \Omega\in \mathcal{F}^{(N)}$, since every addend in Eq. (\ref{eq:Npart}) is an element of $\mathcal{F}^{(N)}$.

The intuition behind the choice Eq. (\ref{eq:trial}) is the following. Recall the expression of the many-body Hamiltonian after conjugation with the Bogoliubov transformation, Eq. (\ref{eq:low1}). The previous analysis suggests that $\mathcal{H}_{N} \simeq \mathcal{E}_{\text{N}}^{\text{HF}}(\omega_{N}) + R_{\omega_{N}} \mathbb{H}_{0} R^{*}_{\omega_{N}} + R_{\omega_{N}} (\mathbb{F} + \mathbb{F}^{*}) R^{*}_{\omega_{N}}$, up to terms that are small if one restricts the Fock space to states with low kinetic energy $\mathbb{H}_{0}$. We are interested in finding an expression for the ground state of this approximate Hamiltonian. The state $\mathbb{R}_{\omega_{N}} \Omega$ is the ground state of the ``unperturbed'' operator $R_{\omega_{N}} \mathbb{H}_{0} R^{*}_{\omega_{N}}$. The question is how does this eigenstate change as one includes the ``perturbation'' $R_{\omega_{N}} (\mathbb{F} + \mathbb{F}^{*}) R^{*}_{\omega_{N}}$. The ansatz (\ref{eq:trial}) is what one would obtain applying first-order perturbation theory, neglecting higher order terms.

Let us now prove Proposition \ref{prp:upper}.
\begin{proof}
Let $\varphi_{\mathbb{F}} = M_{\mathbb{F}}^{-1} ( 1 - \mathbb{H}_{0}^{-1} \mathbb{F}^{*} )\Omega$. Then, by the variational principle, the ground state energy is bounded above by:
\begin{eqnarray}
E_{N} &\leq& \langle \psi_{\mathbb{F}}, \mathcal{H}_{N} \psi_{\mathbb{F}} \rangle\nonumber\\
&\equiv& \langle \varphi_{\mathbb{F}}, R^{*}_{\omega_{N}} \mathcal{H}_{N} R_{\omega_{N}} \varphi_{\mathbb{F}} \rangle\;.
\end{eqnarray}
Eq. (\ref{eq:low1}) together with the estimates in Proposition \ref{prp:boundsQ} and with Corollary \ref{cor:Nimpro} imply:
\begin{eqnarray}
E_{N} &\leq& \mathcal{E}^{\text{HF}}_{N}(\omega_{N}) + (1 + C\| \hat{v} \|_{1})\langle \varphi_{\mathbb{F}}, \mathbb{H}_{0} \varphi_{\mathbb{F}} \rangle + \langle \varphi_{\mathbb{F}}, \mathbb{Q} \varphi_{\mathbb{F}}\rangle + C_{\delta}N^{- \frac{9}{16} + \delta}\nonumber\\
&=& \mathcal{E}^{\text{HF}}_{N}(\omega_{N}) + (1 + C\| \hat{v} \|_{1})\langle \varphi_{\mathbb{F}}, \mathbb{H}_{0} \varphi_{\mathbb{F}} \rangle\nonumber\\
&& + \langle \varphi_{\mathbb{F}}, (\mathbb{Q}_{1} + \mathbb{Q}_{2,a}+ \mathbb{F} + \mathbb{F}^{*}) \varphi_{\mathbb{F}}\rangle + C_{\delta}N^{- \frac{9}{16} + \delta}\nonumber\\
&\leq& \mathcal{E}^{\text{HF}}_{N}(\omega_{N}) + M_{\mathbb{F}}^{-2}(1 + C\| \hat{v} \|_{1}) \langle \Omega, \mathbb{F} \mathbb{H}_{0}^{-1} \mathbb{F}^{*} \Omega \rangle \\
&& + \langle \varphi_{\mathbb{F}}, (\mathbb{Q}_{1} + \mathbb{Q}_{2,a}) \varphi_{\mathbb{F}}\rangle  - 2M_{\mathbb{F}}^{-2}\langle \Omega, \mathbb{F} \mathbb{H}_{0}^{-1} \mathbb{F}^{*} \Omega \rangle + C_{\delta}N^{- \frac{9}{16} + \delta}\;.\nonumber
\end{eqnarray}
By Eq. (\ref{eq:Q2a}), we have:
\begin{eqnarray}
\langle \varphi_{\mathbb{F}}, \mathbb{Q}_{2,a} \varphi_{\mathbb{F}}\rangle &\leq& C\|\hat{v}\|_{1} \langle \varphi_{\mathbb{F}}, \mathbb{H}_{0} \varphi_{\mathbb{F}}\rangle\nonumber\\
&=& C \|\hat{v}\|_{1} \langle \Omega, \mathbb{F} \mathbb{H}_{0}^{-1} \mathbb{F}^{*} \Omega \rangle\;.
\end{eqnarray}
Finally,
\begin{eqnarray}
\langle \varphi_{\mathbb{F}}, \mathbb{Q}_{1} \varphi_{\mathbb{F}}\rangle = \frac{1}{M_{\mathbb{F}}^{2}} \langle\Omega, \mathbb{F} \mathbb{H}_{0}^{-1} \mathbb{Q}_{1} \mathbb{H}_{0}^{-1} \mathbb{F}^{*}\Omega \rangle\;.
\end{eqnarray}
To estimate this term, we use that $\mathbb{Q}_{1} \leq (\| \hat{v} \|_{1}/N) \mathcal{N}^{2}$. Using that $\mathbb{F}^{*}\Omega \in \mathcal{F}^{(4)}$, and that $\langle \Omega, \mathbb{F} \mathbb{H}_{0}^{-2} \mathbb{F}^{*} \Omega \rangle \leq C\|\hat{v}\|_{1}^{2}$ (see Appendix \ref{app:norm}), we get:
\begin{equation}
\langle \varphi_{\mathbb{F}}, \mathbb{Q}_{1} \varphi_{\mathbb{F}}\rangle \leq \frac{16 \| \hat{v} \|_{1}}{N M_{\mathbb{F}}^{2}} \langle \Omega, \mathbb{F} \mathbb{H}_{0}^{-2} \mathbb{F}^{*} \Omega \rangle\leq \frac{C\| \hat{v} \|^{3}_{1}}{N}\;.
\end{equation}
Therefore, all in all we have:
\begin{equation}
E_{N} \leq \mathcal{E}^{\text{HF}}_{N}(\omega_{N}) - (1 - C\| \hat{v} \|_{1}) M_{\mathbb{F}}^{-2} \langle \Omega, \mathbb{F} \mathbb{H}_{0}^{-1} \mathbb{F}^{*} \Omega\rangle + C_{\delta}N^{- \frac{9}{16} + \delta}\;.
\end{equation}
Eq. (\ref{eq:upcor}) follows for $\|\hat{v}\|_{1}$ small enough, after suitably redefining the constant $C$. This concludes the proof of Proposition \ref{prp:upper}.
\end{proof}

\noindent{\bf Acknowledgements.} The work of M. P. has been supported by the Swiss National Science Foundation via the grant ``Mathematical Aspects of Many-Body Quantum Systems''. We thank Niels Benedikter, Phan Th\`anh Nam and Benjamin Schlein for interesting discussions.

\appendix

\section{Proof of Lemma \ref{lem:H0}}\label{app:key}
We start by writing:
\begin{eqnarray}
&&\sum_{k} \| b_{k+p} c_{k} \varphi \| = \sum_{k} \frac{1}{\sqrt{e(k+p) + e(k)}} \sqrt{e(k+p) + e(k)}  \| b_{k+p} c_{k} \varphi \| \nonumber\\
&&\leq \Big(\sum_{k:\, k+p\notin \mathcal{B}_{\mu},\, k\in \mathcal{B}_{\mu}} \frac{1}{e(k+p) + e(k)}\Big)^{\frac{1}{2}} \Big(\sum_{k} (e(k+p) + e(k)) \| b_{k+p} c_{k}  \|^{2}\Big)^{\frac{1}{2}}\nonumber\\
&&\leq N^{\frac{1}{2}} I_{\mu}(p)^{\frac{1}{2}} \| \mathbb{H}_{0}^{\frac{1}{2}} \varphi \|\;,
\end{eqnarray}
where we defined:
\begin{equation}
I_{\mu}(p) = \frac{1}{N} \sum_{k:\, k+p\notin \mathcal{B}_{\mu},\, k\in \mathcal{B}_{\mu}} \frac{1}{e(k+p) + e(k)}\;.
\end{equation}
We claim that, for $|p| \leq N^{\zeta}$ and $\zeta>0$ to be determined below:
\begin{equation}\label{eq:Itot}
I_{\mu}(p) \leq C\;,
\end{equation}
which would conclude the proof of Lemma \ref{lem:H0}. To prove this, we shall split the sum in three terms:
\begin{equation}
I_{\mu}(p) = I^{\text{A}}_{\mu}(p) + I^{\text{B}}_{\mu}(p) + I^{\text{C}}_{\mu}(p)
\end{equation}
with, for $A$ large enough (independent of $|p|$ and $N$), and with $\gamma = 131/208$ (the motivation for this choice will be clear later on):
\begin{eqnarray}
I^{\text{A}}_{\mu}(p) &=& \frac{1}{N} \sum_{\substack{k:\, k+p\notin \mathcal{B}_{\mu},\, k\in \mathcal{B}_{\mu} \\ e(k) + e(k+p) \leq A|p|\eps^{2}}} \frac{1}{e(k+p) + e(k)}\nonumber\\
I^{\text{B}}_{\mu}(p) &=& \frac{1}{N} \sum_{\substack{k:\, k+p\notin \mathcal{B}_{\mu},\, k\in \mathcal{B}_{\mu} \\ A|p|\eps^{2} < e(k) + e(k+p) \leq \eps^{2 - \gamma}}} \frac{1}{e(k+p) + e(k)}\nonumber\\
I^{\text{C}}_{\mu}(p) &=& \frac{1}{N} \sum_{\substack{k:\, k+p\notin \mathcal{B}_{\mu},\, k\in \mathcal{B}_{\mu} \\ \eps^{2-\gamma} < e(k) + e(k+p)}} \frac{1}{e(k+p) + e(k)}\;.
\end{eqnarray}
We shall study the three terms separately. As we shall see, the main contribution to $I_{\mu}(p)$ will come from $I^{\text{C}}_{\mu}(p)$. The other two terms will be proven to be $o(1)$ as $\eps \to 0$.

\medskip

\noindent{\underline{Term $A$}.} Using that $e(k+p) + e(k) \geq \eps^{2}$, we estimate:
\begin{equation}
I^{\text{A}}_{\mu}(p) \leq \frac{1}{N \eps^{2}} |\mathcal{R}_{p}|
\end{equation}
where $\mathcal{R}_{p}$ is the set:
\begin{equation}
\mathcal{R}_{p} = \big\{ k\in \mathcal{B}_{\mu}\mid k+p\notin \mathcal{B}_{\mu},\; e(k+p) + e(k) \leq A|p|^{2}\varepsilon^{2} \big\}\;.
\end{equation}
Recall that $e(k) = |\eps^{2} |k|^{2} - \mu|$. The constraint $e(k) \leq A|p|^{2}\eps^{2}$ is equivalent to $||k|^{2} - N^{\frac{2}{3}} \mu| \leq A|p|^{2}$, that is $||k| - N^{\frac{1}{3}} \sqrt{\mu}|\leq C|p|^{2}N^{-\frac{1}{3}}$: the lattice points in $\mathcal{R}_{p}$ are at a distance at most $C |p|^{2} N^{-\frac{1}{3}}$ from the Fermi surface. Also, the constraints $k\in \mathcal{B}_{\mu}$, $k+p\notin\mathcal{B}_{\mu}$ imply that:
\begin{equation}\label{eq:ee}
e(k+p) + e(k) = \eps^{2} |k+p|^{2} - \mu + \mu - \eps^{2} |k|^{2} = \eps^{2}|p|^{2} + 2\eps^{2}k\cdot p\;.
\end{equation}
Let $\theta_{k}$ the relative angle of $k$ and $p$, $k\cdot p = |k| |p| \cos\theta_{k}$. Since $|k| = O(N^{\frac{1}{3}})$, the constraint $e(k+p) + e(k) \leq A|p|^{2}\varepsilon^{2}$ together with Eq. (\ref{eq:ee}) implies that
\begin{equation}
|\cos\theta_{k}| \leq CA|p| \eps\;.
\end{equation}
Hence, the set $\mathcal{R}_{p}$ has width bounded by $CA|p|$ in the direction of $p$.

To count the lattice points in $\mathcal{R}_{p}$, we slice $\mathcal{R}_{p}$ with planes orthogonal to the vector $p$. Let $\{ e_{1}, e_{2}, e_{3} \}$ be the standard orthogonal basis for $(2\pi) \mathbb{Z}^{3}$, with $|e_{i}| = 2\pi$. If $p$ is parallel to one of the basis vectors, say $e_{1}$, the lattice points intersected by these planes form subsets of $(2\pi)\mathbb{Z}^{2}$. We write:
\begin{eqnarray}
\mathcal{R}_{p} &=& \bigcup_{n} \big\{ k\in \mathcal{R}_{p} \mid k\cdot e_{1} = (2\pi)^{2} n \big\} \nonumber\\
&\equiv& \bigcup_{n} \mathcal{R}_{p, n}\;.
\end{eqnarray} 
Since the width of $\mathcal{R}_{p}$ in the direction of $p$ is bounded by $CA|p|$, to count the number of lattice points in $\mathcal{R}_{p}$ it will be enough to count the number of lattice points in a single slice $\mathcal{R}_{p, n}$. This number is bounded by the number of lattice points lying in the difference of two disks:
\begin{equation}\label{eq:Rpnfirst}
| \mathcal{R}_{p,n} | \leq | \{ k\in (2\pi) \mathbb{Z}^{2}\mid N^{\frac{1}{3}}\nu - C|p|^{2}N^{-\frac{1}{3}} \leq |k| \leq N^{\frac{1}{3}}\nu  \} |\;,
\end{equation}
for some $\nu = \sqrt{\mu} - O(N^{-\frac{1}{3}})$. To estimate the right-hand side of Eq. (\ref{eq:Rpnfirst}) we proceed as follows. Let $\mathcal{C}\subset \mathbb{R}^{2}$ be a convex body with smooth boundary ({\it e.g.}, a disk). Then, for any $\delta > 0$ and for $\gamma = 131/208$, \cite{Hux}:
\begin{equation}\label{eq:count2}
| \{ x\in \mathbb{Z}^{2}\mid x/R \in \mathcal{C}  \} | = R^{2} \mu(\mathcal{C}) + E_{2}(R)\;,\quad |E_{2}(R)| \leq C_{\delta} R^{\gamma + \delta}
\end{equation}
where $\mu(\mathcal{C})$ is the area of $\mathcal{C}$. Therefore, thanks to Eq. (\ref{eq:count2}), it is easy to estimate $|\mathcal{R}_{p,n}|$ as:
\begin{equation}\label{eq:Rpn}
| \mathcal{R}_{p,n} | \leq | \{ k\in (2\pi) \mathbb{Z}^{2}\mid N^{\frac{1}{3}}\nu - CN^{-\frac{1}{3}} \leq |k| \leq N^{\frac{1}{3}}\nu  \} | \leq K_{\delta} N^{\frac{1}{3} (\gamma + \delta)}\;.
\end{equation}
Hence, Eq. (\ref{eq:Rpn}) implies that
\begin{equation}
I^{\text{A}}_{\mu}(p) \leq  \frac{1}{N^{\frac{1}{3}}}|\mathcal{R}_{p}| \leq CA|p| K_{\delta} N^{\frac{1}{3}(\gamma + \delta) - \frac{1}{3}}\;,
\end{equation}
which is $o(1)$ for $|p| \leq N^{\zeta}$ with $\zeta < \frac{1}{3} - \frac{1}{3}(\gamma + \delta)$.

Suppose now that $p$ is not aligned with any basis vector $e_{1}, e_{2}, e_{3}$. We introduce a new orthogonal basis $\{ b_{1}, b_{2}, b_{3} \}$, with $b_{1} = (2\pi) p/|p|$ and $|b_{i}| = 2\pi$. For instance, let:
\begin{equation}
p_{\perp} = (0, -p_{3}, p_{2})\;,\qquad p_{\perp}' = ( -p_{2}^{2} - p_{3}^{2}, p_{1} p_{2}, p_{1} p_{3})\;.
\end{equation}
Then, we set:
\begin{equation}
b_{1} = (2\pi) \frac{p}{|p|}\;,\qquad b_{2} = (2\pi) \frac{p_{\perp}}{|p_{\perp}|}\;,\qquad b_{3} = (2\pi) \frac{p_{\perp}'}{|p_{\perp}'|}\;.
\end{equation}
For later use, notice that $|p_{\perp}'| = |p| |p_{\perp}| \leq |p|^{2}$. We then slice the ribbon $\mathcal{R}_{p}$ with planes orthogonal to $b_{1}$:
\begin{eqnarray}\label{eq:intRp}
\mathcal{R}_{p} &=& \bigcup_{\alpha} \big\{ k\in \mathcal{R}_{p} \mid k\cdot b_{1} = (2\pi)^{2}\alpha \big\} \nonumber\\
&\equiv& \bigcup_{\alpha} \mathcal{R}_{p,\alpha}\;.
\end{eqnarray}
Since $\alpha = \frac{1}{(2\pi)|p|}\sum_{i} k_{i} p_{i}$ with $k, p\in (2\pi)\mathbb{Z}^{3}$, the smallest increment of $\alpha$ is bounded below by $(\text{const.}) |p|^{-1}$. Also, since the width of $\mathcal{R}_{p}$ in the direction of $p$ is bounded by $CA|p|$, the number of sets $\mathcal{R}_{p,\alpha}$ in the union in Eq. (\ref{eq:intRp}) is bounded by $CA|p|^{2}$. To bound the number of lattice points in $\mathcal{R}_{p}$ it is again enough to bound the number of lattice points in a given $\mathcal{R}_{p,\alpha}$.

In contrast to the previous case, the number of lattice points in $\mathcal{R}_{p,\alpha}$ is not a subset of $\mathbb{Z}^{2}$. However, it can be estimated with the number of lattice points in a suitable subset of $\mathbb{Z}^{2}$, proceeding as follows. Any lattice point belonging to $\mathcal{R}_{p,\alpha}$ can be written as $k = \alpha b_{1} + \beta b_{2} + \gamma b_{3}$, with 
\begin{eqnarray}\label{eq:abg}
\alpha &=&  \frac{1}{(2\pi)^{2}} k\cdot b_{1} = \frac{1}{(2\pi)} \frac{k\cdot p}{|p|}\nonumber\\
\beta &=& \frac{1}{(2\pi)^{2}} k\cdot b_{2} = \frac{1}{(2\pi)} \frac{k\cdot p_{\perp}}{|p_{\perp}|} \nonumber\\
\gamma &=& \frac{1}{(2\pi)^{2}} k\cdot b_{3} = \frac{1}{(2\pi)} \frac{k\cdot p'_{\perp}}{|p'_{\perp}|}\;.
\end{eqnarray}
Different points in $\mathcal{R}_{p,\alpha}$ correspond to different values of $\beta$ and $\gamma$, for $\alpha$ fixed. The coefficients $\beta$, $\gamma$ are not integers, in general. However, they become integer valued if properly rescaled:
\begin{equation}
n_{2} = \frac{1}{(2\pi)}|p_{\perp}| \beta \in \mathbb{Z}\;,\qquad  n_{3} = \frac{1}{(2\pi)^{2}}|p'_{\perp}|  \gamma \in \mathbb{Z}\;.
\end{equation}
%
%Let us now parametrize those points in terms of $k_{1}$ and $k_{2}$. In fact, the condition $k\cdot b_{1} = (2\pi)^{2}\alpha$ can be used to fix $k_{3}$ a function of $k_{1}$ and $k_{2}$ (since $p_{3}\neq 0$ by assumption). We get:
%
%\begin{equation}
%k_{3} = \frac{1}{p_{3}} ( |p| \alpha - k_{1} p_{1} - k_{2} p_{2})\;.
%\end{equation}
%
%Therefore,
%
%\begin{eqnarray}\label{eq:betagamma}
%\beta &=& \frac{1}{(2\pi)^{2}}\Big( k_{1} b_{2,1} + k_{2} b_{2,2} + k_{3} b_{2,3} \Big) \\
%&=& \frac{1}{(2\pi)^{2}}\Big( k_{1} b_{2,1} + k_{2} b_{2,2} + b_{2,3} \frac{|p|\alpha}{p_{3}} - b_{2,3} \frac{k_{1} p_{1}}{p_{3}} - b_{2,3} \frac{k_{2} p_{2}}{p_{3}} \Big)\nonumber\\
%\gamma &=& \frac{1}{(2\pi)^{2}} \Big( k_{1} b_{3,1} + k_{2} b_{3,2} + k_{3} b_{3,3} \Big) \nonumber\\
%&=& \frac{1}{(2\pi)^{2}}\Big( k_{1} b_{3,1} + k_{2} b_{3,2} + b_{3,3} \frac{|p|\alpha}{p_{3}} - b_{3,3} \frac{k_{1} p_{1}}{p_{3}} - b_{3,3} \frac{k_{2} p_{2}}{p_{3}} \Big)\;.\nonumber
%\end{eqnarray}
%
%Thus, $\beta$ and $\gamma$ have the form:
%
%\begin{equation}
%\beta = \frac{1}{(2\pi)^{2}} (k_{1} c_{1} + k_{2} c_{2} + c_{3})\;,\qquad \gamma = \frac{1}{(2\pi)^{2}} ( k_{1} d_{1} + k_{2} d_{2} + d_{3})\;,
%\end{equation}
%
%for some $k$-independent coefficients $c_{i}, d_{i}$, not in $(2\pi)\mathbb{Z}$ in general. However, notice that:
%
%\begin{eqnarray}
%p_{3} |p_{\perp}| |p_{\perp}'| \beta &=& \frac{1}{(2\pi)^{2}} (k_{1} C_{1} + k_{2} C_{2} + C_{3})\;,\nonumber\\
%p_{3} |p_{\perp}| |p_{\perp}'| \gamma &=& \frac{1}{(2\pi)^{2}} (k_{1} D_{1} + k_{2} D_{2} + D_{3})
%\end{eqnarray}
%
%where $C_{i}, D_{i}$ are in $(2\pi) \mathbb{Z}$. 
Therefore, the number of lattice points in $\mathcal{R}_{p,\alpha}$ is estimated by the number of lattice points in:
\begin{eqnarray}\label{eq:n1n2}
&&\Big\{ (n_{2}, n_{3}) \in \mathbb{Z}^{2}\, \Big|\, N^{\frac{2}{3}} \tilde \nu - C|p|^{2}\leq \frac{(2\pi)^{4}}{|p_{\perp}|^{2}} n_{2}^{2} + \frac{(2\pi)^{6}}{|p_{\perp}'|^{2}} n_{3}^{2} \leq N^{\frac{2}{3}} \tilde \nu\Big\} \\
&& = \Big\{ (n_{2}, n_{3}) \in \mathbb{Z}^{2}\, \Big|\, K_{p} N^{\frac{2}{3}} \tilde \nu - C_{p}|p|^{2}\leq \frac{|p|^{2}}{(2\pi)^{2}} n_{2}^{2} +  n_{3}^{2} \leq K_{p} N^{\frac{2}{3}} \tilde \nu\Big\}\nonumber
\end{eqnarray}
with $N^{\frac{2}{3}} \tilde \nu = N^{\frac{2}{3}} \nu - \alpha^{2} (2\pi)^{2}$, and where $K_{p} = |p_{\perp}'|^{2}/(2\pi)^{6}$, $C_{p} = C |p_{\perp}'|^{2}/(2\pi)^{6}$. The number of lattice points in (\ref{eq:n1n2}) is less or equal than the number of elements of:
\begin{equation}\label{eq:ellipses}
\Big\{ (n_{2}, n_{3}) \in \mathbb{Z}^{2}\, \Big|\, K_{p} N^{\frac{2}{3}} \tilde \nu - C_{p}|p|^{2}\leq \frac{|p|^{2}}{(2\pi)^{2}P^{2}} n_{2}^{2} +  n_{3}^{2} \leq K_{p} N^{\frac{2}{3}} \tilde \nu\Big\}
\end{equation}
where $P = \lceil |p| \rceil$. This is the number of lattice points lying in the difference of two smooth convex bodies, with elliptic boundary. To estimate the number of lattice points in (\ref{eq:ellipses}), we use again Eq. (\ref{eq:count2}) (thanks to the factor $P$, the curvature of the ellipses is bounded uniformly in $|p|$). We get, for all $\delta > 0$ and for $p$-independent constants $K$ and $C_{\delta}$:
\begin{eqnarray}\label{eq:Rpalpha}
|\mathcal{R}_{p,\alpha}| &\leq& K |p|^{6} + C_{\delta} ( |p|^{2} N^{\frac{1}{3}} \tilde \nu )^{\gamma + \delta} \nonumber\\
&\leq& \widetilde C_{\delta} ( |p|^{2} N^{\frac{1}{3}} \tilde \nu )^{\gamma + \delta}\;,
\end{eqnarray}
for $|p| \leq N^{\zeta}$, with $\zeta < \frac{\gamma + \delta}{18 - 6(\gamma + \delta)}$. Recalling that the number of sets in the union in Eq. (\ref{eq:intRp}) is bounded by $C A|p|^{2}$, we get:
\begin{eqnarray}\label{eq:IAbd}
I^{A}_{\mu}(p) \leq \frac{1}{N^{\frac{1}{3}}}|\mathcal{R}_{p}| &\leq& K_{\delta} N^{\frac{1}{3}(\gamma + \delta) - \frac{1}{3}} |p|^{2 + 2(\gamma + \delta)}\;,
\end{eqnarray}
which is $o(1)$ for $|p| \leq N^{\zeta}$ with $\zeta < \frac{1}{1+ \gamma + \delta} ( \frac{1}{6} - \frac{1}{6}(\gamma + \delta) )$.

\medskip

\noindent{\underline{Term $B$}.} Let us consider the term $I^{B}_{\mu}(p)$. We can parametrize every point $k$ by an angle $\theta_{k} \in [0,\pi]$ such that $k\cdot p = |k| |p| \cos \theta_{k}$, and a remaining angle $\phi_{k} \in [0, 2\pi]$. We rewrite the argument of the sum as:
\begin{eqnarray}
\frac{1}{e(k+p) + e(k)} &=& \frac{1}{\eps^{2} |p|^{2} + 2\varepsilon^{2} k\cdot p}\nonumber\\
&=& \frac{1}{\eps^{2} |p|^{2} + 2\varepsilon^{2} |k| |p| \cos\theta_{k}}\;.
\end{eqnarray}
For $A$ large enough, the constraints $A|p|^{2}\varepsilon^{2} \leq e(k+p) + e(k)\leq \eps^{2-\gamma}$ and the fact that $|k| = O(N^{\frac{1}{3}})$ imply that $a\eps \leq \cos \theta_{k} \leq c\eps^{1-\gamma}$ for suitable $a, c>0$. Using that $||k| - N^{\frac{1}{3}} \sqrt{\mu}| \leq C$ for all points in the sum, we estimate:
\begin{equation}
\frac{1}{e(k+p) + e(k)} \leq \frac{C}{\eps^{2}|p|^{2} + 2\varepsilon \sqrt{\mu}  |p| \cos\theta_{k}}\;.
\end{equation}
Let us now introduce a family of angles $\{\eta_{i}\}$ partitioning $[0, \pi/2]$, such that
\begin{equation}
|\eta_{i} - \eta_{i+1}| = \kappa \eps\;,
\end{equation}
for some $\kappa >0$ of order $1$. We write:
\begin{eqnarray}
I^{\text{B}}_{\mu}(p) &\leq& \frac{1}{N} \sum_{a\eps \leq \cos \eta_{i} \leq c\eps^{1-\gamma}} \sum_{\substack{ k:\, k+p\notin \mathcal{B}_{\mu},\, k\in \mathcal{B}_{\mu} \nonumber\\ A|p|^{2}\varepsilon^{2} \leq e(k+p) + e(k) \leq \varepsilon^{2 - \gamma} \\ |\theta_{k} -  \eta_{i}| \leq \frac{\kappa}{2}\eps}} \frac{C}{\eps^{2}|p|^{2} + 2\varepsilon \sqrt{\mu}  |p| \cos\theta_{k}} \nonumber\\
&\leq&  \frac{1}{N} \sum_{a\eps \leq \cos \eta_{i} \leq c\eps^{1-\gamma}}  \frac{D}{\eps^{2}|p|^{2} + 2\varepsilon \sqrt{\mu}  |p| \cos\eta_{i}} \sum_{\substack{ k:\, k+p\notin \mathcal{B}_{\mu},\, k\in \mathcal{B}_{\mu} \nonumber\\  |\theta_{k} - \eta_{i}| \leq \frac{\kappa}{2}\eps}} 1\;.
\end{eqnarray}
The innermost sum is the number of lattice points in:
\begin{equation}
\mathcal{S}_{p,i} = \Big\{ k\in \mathcal{B}_{\mu} \mid k+p\notin \mathcal{B}_{\mu},\;  | \theta_{k} - \eta_{i} |\leq \frac{\kappa}{2}\eps  \Big\}\;.
\end{equation}
To estimate the number of points in this set, we proceed as in the discussion of the term $I^{A}_{\mu}(p)$. We introduce the orthogonal basis $\{ b_{1}, b_{2}, b_{3} \}$ as in Eq. (\ref{eq:abg}), and we slice the set $\mathcal{S}_{p,i}$ with planes with fixed $b_{1}$ component; we then estimate the number of lattice points in each slice, and use that $\mathcal{S}_{p,i}$ is sliced by a number of planes bounded by $C|p|$ (since the width of $\mathcal{S}_{p,i}$ is the direction of $p$ is bounded by a constant, and since the spacing between two planes is bounded below by $C/|p|$). We write:
\begin{eqnarray}\label{eq:Spi}
\mathcal{S}_{p,i} &=& \bigcup_{\alpha} \mathcal{S}_{p,i,\alpha} \nonumber\\
\mathcal{S}_{p,i,\alpha} &=& \big\{ k\in \mathcal{S}_{p,i} \mid k\cdot b_{1} = (2\pi)^{2}\alpha \big\}\;.
\end{eqnarray}
A simple trigonometric argument shows that the conditions $k\in \mathcal{B}_{\mu}$, $k+p\notin \mathcal{B}_{\mu}$, together with $| \cos\theta_{k} - \cos\eta_{i} |\leq C\kappa\eps$ and $\cos \eta_{i} \geq a\eps$, imply that the width of the collar $\mathcal{S}_{p,i,\alpha}$ in any direction orthogonal to $p$ is bounded above by $|p| \cos\eta_{i}$. 

Proceeding as in the analysis of $I^{A}_{\mu}(p)$, following the steps to arrive at Eq. (\ref{eq:Rpalpha}), the number of lattice points in $\mathcal{S}_{p,i,\alpha}$ is bounded by:
\begin{equation}
K |p|^{6} N^{\frac{1}{3}}\cos \eta_{i} + C_{\delta} ( |p|^{2} N^{\frac{1}{3}} \tilde \nu )^{\gamma + \delta}\;.
\end{equation}
Since $\cos \eta_{i} \leq c\eps^{1-\gamma}$, the second term dominates for $|p| \leq N^{\zeta}$ with $\zeta = \frac{\delta}{18 - 6(\gamma + \delta)}$. We then have, taking into account that the number of sets in the union in Eq. (\ref{eq:Spi}) is bounded by $C|p|$:
\begin{equation}
I^{\text{B}}_{\mu}(p) \leq \frac{K}{N^{\frac{1}{3}}} \sum_{i: a\eps \leq \cos \eta_{i} \leq c \eps^{1-\gamma-\delta}} \frac{|p|^{2(\gamma + \delta) + 1} \eps^{1-\gamma -\delta} }{\eps |p|^{2} + 2 \sqrt{\mu}  |p| \cos\eta_{i}}\;.
\end{equation}
This sum can be estimated as:
\begin{eqnarray}\label{eq:IBest2}
I^{\text{B}}_{\mu}(p) &\leq& \widetilde{K} |p|^{2(\gamma + \delta)} \int_{0}^{\frac{\pi}{2}} dx\, \frac{\eps^{1-\gamma - \delta}}{\eps + \cos x}\nonumber\\
&\leq& C |p|^{2(\gamma + \delta)} \eps^{1-\gamma - \delta} |\log \eps|\;,
\end{eqnarray}
which is $o(1)$ for $|p| \leq N^{\zeta}$ with $\zeta < \frac{1-\gamma -\delta}{6(\gamma + \delta)}$. 
\medskip

\noindent{\underline{Term C}.} To conclude, let us consider $I^{\text{C}}_{\mu}(p)$. Up to $o(1)$ error terms, we can replace the third constraint in the sum by:
\begin{equation}\label{eq:C0}
C\varepsilon^{2-\gamma} < \eps^{2} |k| |p| \cos\theta_{k}\;;
\end{equation}
again up $o(1)$ error terms, we can replace Eq. (\ref{eq:C0}) by:
\begin{equation}
\frac{C}{|p|} \varepsilon^{1-\gamma} < \cos\theta_{k}\;.
\end{equation}
Next, consider the summand. We replace $|k|$ with $N^{\frac{1}{3}}\sqrt{\mu}$, up to a small error:
\begin{equation}\label{eq:rk}
\frac{1}{e(k) + e(k+p)} = \frac{1}{\eps^{2} |p|^{2} + 2\eps\sqrt{\mu} |p| \cos\theta_{k}} (1 + r_{k,p})
\end{equation}
with $|r_{k,p}| \leq C\eps |p|$. Then, we introduce angles $\xi_{i}$ partitioning $[0,\pi/2]$ such that:
\begin{equation}
|\xi_{i} - \xi_{i+1}| = \kappa \eps^{1 - \gamma + \alpha}
\end{equation}
for $\alpha>0$ small enough, to be chosen later, and $\kappa = O(1)$. For $|\theta_{k} - \xi_{i}| \leq \frac{\kappa}{2} \eps^{1-\gamma + \alpha}$, we further rewrite the summand as:
\begin{equation}\label{eq:r2k}
\frac{1}{\eps^{2} |p|^{2} + 2\eps\sqrt{\mu} |p| \cos\theta_{k}} = \frac{1}{\eps^{2} |p|^{2} + 2\eps\sqrt{\mu} |p| \cos\xi_{i}} (1 + r'_{k,p})\;,
\end{equation}
where $|r'_{k,p}| \leq C\eps^{\alpha}$. Neglecting for the moment the error terms, we are left with computing:
\begin{equation}\label{eq:Cmain}
\frac{1}{N} \sum_{\frac{C}{|p|}\eps^{1-\gamma} \leq \cos \xi_{i}}\frac{1}{\eps^{2} |p|^{2} + 2\eps\sqrt{\mu} \cos\xi_{i}} \sum_{\substack{ k:\, k+p\notin \mathcal{B}_{\mu},\, k\in \mathcal{B}_{\mu} \\  |\theta_{k} - \xi_{i}| \leq \frac{\kappa}{2}\eps^{1-\gamma + \alpha}}} 1\;.
\end{equation}
This sum can be bounded by a constant, uniformly in $\eps$, repeating the slicing and counting argument used for $I^{B}_{\mu}(p)$. We shall however adopt a different strategy, that will allow us to explicitly compute (\ref{eq:Cmain}) as $\eps \to 0$. The reason is that the same argument can be used to prove Eq. (\ref{eq:cormain}), about the $\eps\to 0$ limit of the correlation energy at second order in the potential, see Appendix \ref{app:corr}.

The innermost sum can be studied as in Section 6 of \cite{BNPSS}. With respect to \cite{BNPSS}, the only difference is that instead of counting lattice points in a patch around the Fermi surface, Eq. (6.1) of \cite{BNPSS}, we are interested in the number of lattice points in a ``collar'':
\begin{equation}
\mathcal{C}_{i} = \Big\{ k\in \mathcal{B}_{\mu} \mid k+p\notin \mathcal{B}_{\mu},\; | \theta_{k} - \xi_{i}| \leq \frac{\kappa}{2}\eps^{1-\gamma + \alpha} \Big\}\;.
\end{equation}
Let us denote by $C_{i}$ the corresponding subset of $N^{1/3} \sqrt{\mu}\, \mathbb{S}^{2}$:
\begin{equation}
C_{i} = \{ x\in N^{\frac{1}{3}} \sqrt{\mu}\, \mathbb{S}^{2} \mid |\theta - \xi_{i}| \leq  \frac{\kappa}{2}\eps^{1-\gamma + \alpha}  \}\;,
\end{equation}
where $\theta$ is such that $x\cdot p = |x| |p| \cos \theta$. Let us denote by $C_{i,p}$ the projection of $C_{i}$ along $p$ onto $\mathbb{R}^{2}\times \{0\}$, see Fig. 6 of \cite{BNPSS}. Without loss of generality, we are assuming that $p_{3} \neq 0$; otherwise, we pick another component of $p$, and project over another coordinate plane. Then, see Eq. (6.4) of \cite{BNPSS}:
\begin{equation}\label{eq:muC}
|\mathcal{C}_{i}| = \mu(C_{i,p}) p_{3} + \frak{e}_{i,p}\;,
\end{equation}
where $\mu$ is the Lebesgue measure on the plane and the error term $\frak{e}_{i,p}$ is bounded proportionally to $|p| |\partial C_{i,p}|$, which is estimated by $C |p|^{2} N^{\frac{1}{3}} \sin \xi_{i}$, see discussion after Eq. (6.4) of \cite{BNPSS}. Finally, the main term in Eq. (\ref{eq:muC}) is, see Eq. (6.3) of \cite{BNPSS}:
\begin{equation}
\mu(C_{i,p}) = N^{\frac{2}{3}} \mu \frac{|p|}{p_{3}} \cos \xi_{i} (1 + O(\eps^{\alpha})) \sigma (c_{i})\;,
\end{equation}
where $c_{i}$ is the subset of the unit sphere $\mathbb{S}^{2}$, such that $N^{\frac{1}{3}}\sqrt{\mu} c_{i} = C_{i}$ and $\sigma(c_{i})$ is the surface area of $c_{i}$. Referring to the notations of Eq. (6.3) of \cite{BNPSS}, the factor $M^{-1/2}$ of \cite{BNPSS} is here replaced by $\eps^{1-\gamma + \alpha}$ (corresponding to the largest variation of the angles $\theta_{k}$ with respect to $\xi_{i}$), and the factor $N^{-\delta}$ of \cite{BNPSS} is here replaced by $\eps^{1-\gamma}$ (corresponding to the lower bound for $\cos \xi_{i}$).

%Using that $\gamma = 131/208$, $\eps^{2\gamma - \alpha - 1 } = o(1)$ for $\alpha$ small enough.

Consider the first term in the right-hand side of Eq. (\ref{eq:muC}). The corresponding contribution to $I^{C}_{\mu}(p)$ is:
\begin{equation}\label{eq:bnpss}
\mu |p| \sum_{\frac{C}{|p|}\eps^{1-\gamma} \leq \cos \xi_{i}}\frac{\cos \xi_{i}}{\eps |p|^{2} + 2\sqrt{\mu} |p|\cos\xi_{i}} \sigma(c_{i})\;,
\end{equation}
which converges to a bounded integral as $\eps\to 0$:
\begin{eqnarray}\label{eq:Iint}
(\ref{eq:bnpss}) &=& 2\pi \mu \int_{0}^{\pi/2} d\xi\, \frac{\sin \xi \cos \xi}{\eps |p| + 2\sqrt{\mu} \cos\xi}(1 + o(1)) \nonumber\\
|(\ref{eq:bnpss})| &\leq& C\;.
\end{eqnarray}
To conclude, consider now the contribution to $I^{C}_{\mu}(p)$ due to the second term in the right-hand side of Eq. (\ref{eq:muC}). It is bounded as, using that $\sigma(c_{i}) = O(\eps^{1 - \gamma + \alpha} \sin \xi_{i})$:
\begin{eqnarray}
&&\frac{C|p|^{2} N^{\frac{1}{3}}}{N^{\frac{2}{3}}}\sum_{\frac{\widetilde C}{|p|}\eps^{1-\gamma} \leq \cos \xi_{i}}\frac{\sin \xi_{i}}{\eps |p|^{2} + 2\sqrt{\mu}|p| \cos\xi_{i}} \\
&&\qquad \leq K |p|^{2} N^{-\frac{1}{3} + \frac{1}{3}( 1 - \gamma + \alpha )} \sum_{\frac{\widetilde C}{|p|}\eps^{1-\gamma} \leq \cos \xi_{i}}\frac{1}{\eps |p|^{2} + 2\sqrt{\mu} |p|\cos\xi_{i}} \sigma(c_{i})\nonumber\\
&&\qquad \leq \widetilde{K} |p|^{2} N^{\frac{1}{3}(-\gamma + \alpha)} |\log \eps|\;,
\end{eqnarray}
where in the last step we bounded the sum with the corresponding integral. Thus, the second term in Eq. (\ref{eq:muC}) gives rise to a $o(1)$ contribution to $I^{C}_{\mu}(p)$, for $|p| \leq N^{\zeta}$ with $\zeta < \frac{1}{6}(\gamma - \alpha)$.

Finally, it is easy to see that the error terms due to $r_{k,p}$ and $r'_{k,p}$ in Eqs. (\ref{eq:rk}), (\ref{eq:r2k}) give rise to $o(1)$ contributions to $I_{\mu}(p)$. All in all, for $|p| \leq N^{\zeta}$:
\begin{equation}\label{eq:IC}
I^{C}_{\mu}(p) \leq C\;.
\end{equation}
\noindent\underline{Conclusion.} Putting together Eqs. (\ref{eq:IAbd}), (\ref{eq:IBest2}), (\ref{eq:IC}) we see that, for $|p| \leq N^{\zeta}$ with $\zeta = \frac{\delta}{18 - 6(\gamma + \delta)}$ and $\delta>0$ small enough, the main contribution to $I_{\mu}(p)$ is given by the integral in Eq. (\ref{eq:Iint}),
\begin{equation}
I_{\mu}(p) = 2\pi \mu \int_{0}^{\pi/2} d\xi\, \frac{\sin \xi \cos \xi}{\eps |p| + 2\sqrt{\mu} \cos\xi}(1 + o(1))\;.
\end{equation}
In particular, $I_{\mu}(p) \leq C$. This concludes the proof of Lemma \ref{lem:H0}.\qed

\section{Explicit computations}\label{app:A}

\subsection{Computation of $\mathcal{C}^{(2)}_{N}$ in the $N\to \infty$ limit}\label{app:corr}
In this appendix we shall compute $\mathcal{C}_{N}^{(2)}$ in the $N\to \infty$ limit, and we shall prove Eq. (\ref{eq:cormain}). Arguing as in the proof of Lemma \ref{lem:H0}, one gets, for $|p|\leq N^{\zeta}$ with $\zeta = \frac{\delta}{18 - 6(\gamma + \delta)}$ and $\delta > 0$ small enough:
\begin{eqnarray}
&&\frac{1}{2N^{2}} \sum_{\substack{ k,\, k'\in \mathcal{B}_{\mu} \\ k+p,\, k'-p \in \mathcal{B}_{\mu}^{c} }} \frac{1}{e(k+p) + e(k) + e(k'-p) + e(k')} \\
&&\simeq \frac{1}{4}\eps \mu^{2} |p| (2\pi)^{2} \int_{0}^{\frac{\pi}{2}} d\theta  \int_{0}^{\frac{\pi}{2}} d\theta'  \sin \theta \sin \theta'  \frac{\cos\theta \cos\theta'}{\eps |p| + \sqrt{\mu} (\cos\theta + \cos\theta')}\;,\nonumber
\end{eqnarray}
up to $o(1)$ as $N\to \infty$. Let us compute the integral. We rewrite it as:
\begin{equation}\label{eq:corrint}
\frac{1}{4}\eps \mu^{3/2} |p| (2\pi)^{2} \int_{0}^{1} dx \int_{0}^{1}dy\, \frac{x y}{a + x + y}\;,
\end{equation}
where $a = \eps|p| / \sqrt{\mu}$. Explicit evaluation of the integral gives:
\begin{equation}
(\ref{eq:corrint}) = \eps (2\pi)^{3} |p| \frac{\pi}{2} (1 - \log 2) + o(\eps)\;.
\end{equation}
Therefore,
\begin{eqnarray}
\mathcal{C}_{N}^{(2)} &=& -\frac{1}{2 N^{2}} \sum_{p}\sum_{\substack{k, k':\, k, k'\in \mathcal{B}_{\mu} \\ k + p, k'-p \in \mathcal{B}_{\mu}^{c}}} \frac{\hat{v}(p)^{2}}{e(k+p) + e(k) + e(k'-p) + e(k')} + o(\eps) \nonumber\\
&=& -\eps (2\pi)^{3} \frac{\pi}{2} (1 - \log 2) \sum_{p} |p| \hat v_{<}(p)^{2} + o(\eps)\nonumber\\
&=& -\eps (2\pi)^{3} \frac{\pi}{2} (1 - \log 2) \sum_{p} |p| \hat v(p)^{2} + o(\eps)\;,
\end{eqnarray}
where we used the assumption on the decay of $\hat v(p)$. This concludes the check of Eq. (\ref{eq:cormain}).

\subsection{Normalization of the trial state}\label{app:norm}

Here we shall prove that $|M_{\mathbb{F}} - 1|\leq C\|\hat{v}\|_{1}^{2}$, Eq. (\ref{eq:normbd}). We have:
\begin{eqnarray}
M_{\mathbb{F}}^{2} &=& 1 + \langle \Omega, \mathbb{F} \mathbb{H}_{0}^{-2} \mathbb{F}^{*} \Omega \rangle \\
&=& 1 + \frac{1}{4 N^{2}} \sum_{p, k, k'} \sum_{q, r, r'} \hat{v}(p) \hat{v}(q) \langle \Omega, c_{k} c_{k'} b_{k'-p} b_{k+p} b^{*}_{r+q} b^{*}_{r'-q} c^{*}_{r'}  c^{*}_{r}\nonumber\\
&&\cdot  \frac{1}{(\mathbb{H}_{0} + e(r+q) + e(r'-q) + e(r) + e(r'))^{2}}\Omega \rangle\nonumber\\
&=& 1+ \frac{1}{2 N^{2}} \sum_{p}\sum_{\substack{k, k':\, k, k'\in \mathcal{B}_{\mu} \\ k + p, k'-p \in \mathcal{B}_{\mu}^{c}}} \frac{\hat{v}(p)^{2}}{( e(k+p) + e(k) + e(k'-p) + e(k') )^{2}}\nonumber\\
&& -\frac{1}{2N^{2}} \sum_{p}\sum_{\substack{k, k':\, k, k'\in \mathcal{B}_{\mu} \\ k + p, k'-p \in \mathcal{B}_{\mu}^{c}}} \frac{\hat{v}(p) \hat{v}(p - k'+k)}{( e(k+p) + e(k) + e(k'-p) + e(k') )^{2}}\nonumber\\
&\equiv& 1 + J_{1} + J_{2}\;.
\end{eqnarray}
Therefore, for $N$ large enough:
\begin{eqnarray}
J_{1} &\leq& \frac{1}{2 N^{2}} \sum_{p} \hat{v}(p)^{2} \sum_{\substack{k, k':\, k, k'\in \mathcal{B}_{\mu} \\ k + p, k'-p \in \mathcal{B}_{\mu}^{c}}} \frac{1}{e(k+p) + e(k)} \frac{1}{e(k'-p) + e(k')}\nonumber\\
&\equiv& \frac{1}{2} \sum_{p} \hat{v}(p)^{2} I_{\mu}(p) I_{\mu}(-p) \leq C\|\hat v\|_{1}^{2} + \frak{e}_{K}\;,
\end{eqnarray}
with $C$ independent of $N$. In the last step we used the bound $I_{\mu}(p)\leq C$ for $|p|\leq N^{\zeta}$, proven in Appendix \ref{app:key}, and the quantity $\frak{e}_{K}$ takes into account the large $|p|$ dependence of $\hat v(p)$, recall Eqs. (\ref{eq:Qs}). Similarly,
\begin{eqnarray}
|J_{2}| &\leq& \frac{\|\hat{v}\|_{\infty}}{2N^{2}} \sum_{p} \hat{v}(p) \sum_{\substack{k, k':\, k, k'\in \mathcal{B}_{\mu} \\ k + p, k'-p \in \mathcal{B}_{\mu}^{c}}} \frac{1}{e(k+p) + e(k)} \frac{1}{e(k'-p) + e(k')}\nonumber\\
&\leq & \frac{\|\hat{v}\|_{\infty}}{2} \sum_{p} \hat{v}(p) I_{\mu}(p) I_{\mu}(-p) \leq C\|\hat v\|_{1}^{2} + \frak{e}_{K}\;.
\end{eqnarray}
This concludes the check of the boundedness of the normalization $M_{\mathbb{F}}$.
\section{Proof of Eq. (\ref{eq:interaction})}\label{app:interaction}
Here we report the details of the computation leading to Eq. (\ref{eq:interaction}). We have:
\begin{eqnarray}\label{eq:proofint}
&&R^{*}_{\omega_{N}} a^{*}_{k+p} a^{*}_{k'-p} a_{k'} a_{k} R_{\omega_{N}} = (b^{*}_{k+p} + c_{k+p}) (b^{*}_{k' - p} + c_{k'-p})(b_{k'} + c^{*}_{k'})(b_{k} + c^{*}_{k})\nonumber\\
&& = (b^{*}_{k+p}b^{*}_{k'-p} + b^{*}_{k+p} c_{k'-p} + c_{k+p} b^{*}_{k'-p} + c_{k+p} c_{k'-p} )\nonumber\\
&&\qquad\cdot (b_{k'} b_{k} + b_{k'} c^{*}_{k} + c^{*}_{k'} b_{k} + c^{*}_{k'} c^{*}_{k})\;.
\end{eqnarray}
We take the products and put the operators in normal order, using that the only nontrivial anticommutators are:
\begin{equation}
\{ b^{*}_{k}, b_{k'} \} = \delta_{k,k'} \chi(k \notin \mathcal{B}_{\mu})\;,\qquad \{ c^{*}_{k}, c_{k'} \} = \delta_{k,k'}\chi(k\in \mathcal{B}_{\mu})\;.
\end{equation}
We get, collecting together terms that are equal after performing $k\leftrightarrow k'$ and $p\to -p$ (recall that $\hat{v}(p) = \hat{v}(-p)$ by assumption):
\begin{equation}
(\ref{eq:proofint}) = \text{(quartic terms)} + \text{(quadratic terms)} + \text{(constant terms);}
\end{equation}
the quartic terms are:
\begin{eqnarray}\label{eq:quartic}
&&\text{(quartic terms)} = b^{*}_{k+p} b^{*}_{k'-p} b_{k'} b_{k} - 2 b^{*}_{k+p} b^{*}_{k'-p} c^{*}_{k} b_{k'} + b^{*}_{k+p} b^{*}_{k'-p} c^{*}_{k'} c^{*}_{k}\nonumber\\
&&\quad + 2 b^{*}_{k+p} c_{k' - p} b_{k'} b_{k} + 2 b^{*}_{k+p} c^{*}_{k} c_{k'-p} b_{k'} - 2 b^{*}_{k+p} c^{*}_{k'} c_{k'-p} b_{k}\nonumber\\
&&\quad + 2 b^{*}_{k+p} c^{*}_{k'} c^{*}_{k} c_{k'-p} + c_{k+p} c_{k'-p} b_{k'} b_{k} - 2 c^{*}_{k} c_{k+p} c_{k'-p} b_{k'}\nonumber\\
&&\quad + c^{*}_{k'} c^{*}_{k} c_{k+p} c_{k'-p}\;,\nonumber
\end{eqnarray}
the quadratic terms are:
\begin{eqnarray}
&&\text{(quadratic terms)} = 2 b^{*}_{k'-p} b_{k'}\delta_{k+p, k}\chi(k\in \mathcal{B}_{\mu}) - 2 b^{*}_{k+p} b_{k'}\delta_{k, k'-p}\chi(k\in \mathcal{B}_{\mu})\nonumber\\
&&\qquad\qquad - 2 c^{*}_{k} c_{k+p} \delta_{k'-p,k'} \chi(k'\in \mathcal{B}_{\mu}) + 2c^{*}_{k'} c_{k+p}\delta_{k'-p, k} \chi(k \in \mathcal{B}_{\mu})\;,
\end{eqnarray}
while the constant terms are:
\begin{eqnarray}
\text{(constant terms)} &=& \delta_{k'-p, k'}\delta_{k+p, k}\chi(k'\in \mathcal{B}_{\mu})\chi(k\in \mathcal{B}_{\mu})\nonumber\\&& - \delta_{k+p, k'} \delta_{k'-p, k}\chi(k'\in \mathcal{B}_{\mu})\chi(k\in \mathcal{B}_{\mu})\;.
\end{eqnarray}
In performing this computation we used that $b_{k}c_{k} = 0$, to cancel some quadratic terms produced after the normal ordering. Consider Eq. (\ref{eq:quartic}). We get, separating the particle-number conserving terms from the rest:
\begin{eqnarray}\label{eq:Qbb}
&&\frac{1}{2N} \sum_{k, k', p} \hat{v}(p)\times \text{(quartic terms)} \nonumber\\
&&= \frac{1}{2N} \sum_{k, k', p} \hat{v}(p)\big( b^{*}_{k+p} b^{*}_{k'-p} b_{k'} b_{k} + c^{*}_{k'} c^{*}_{k} c_{k+p} c_{k'-p} + 2 b^{*}_{k+p} c^{*}_{k} c_{k'-p} b_{k'}\nonumber\\&&\quad - 2 b^{*}_{k+p} c^{*}_{k'} c_{k'-p} b_{k}\big)\\
&& + \frac{1}{2N} \sum_{k,k',p} \hat{v}(p) \big( - 2 b^{*}_{k+p} b^{*}_{k'-p} c^{*}_{k} b_{k'} + b^{*}_{k+p} b^{*}_{k'-p} c^{*}_{k'} c^{*}_{k} + 2 b^{*}_{k+p} c^{*}_{k'} c^{*}_{k} c_{k'-p}\big)\nonumber\\
&& + \frac{1}{2N} \sum_{k,k', p} \hat{v}(p) \big( 2 b^{*}_{k+p} c_{k' - p} b_{k'} b_{k} + c_{k+p} c_{k'-p} b_{k'} b_{k}  - 2 c^{*}_{k} c_{k+p} c_{k'-p} b_{k'}\big)\;.\nonumber
\end{eqnarray}
Notice that the last two terms are one the adjoint of the other. Eq. (\ref{eq:Qbb}) reproduces $\mathbb{Q}$ in Eq. (\ref{eq:Q}). Consider now the quadratic terms. We get:
\begin{eqnarray}
&&\frac{1}{2N} \sum_{k, k', p} \hat{v}(p) \times \text{(quadratic terms)}\nonumber\\
&& = \frac{1}{N} \sum_{k, k' \in \mathcal{B}_{\mu}} \hat{v}(0) (b^{*}_{k} b_{k} - c^{*}_{k} c_{k}) - \frac{1}{N} \sum_{p} \hat{v}(p) \sum_{k\in \mathcal{B}_{\mu}} (b^{*}_{k+p} b_{k+p} - c^{*}_{k+p} c_{k+p})\nonumber\\
&& \equiv \mathbb{D} + \mathbb{X}\;,
\end{eqnarray}
with $\mathbb{X}$ defined in Eq. (\ref{eq:Q}), and $\mathbb{D}$ defined in Eq. (\ref{eq:Direct}) (recall that $\sum_{k \in \mathcal{B}_{\mu}} = N$). Finally, consider the constant terms. We get:
\begin{eqnarray}
&& \frac{1}{2N} \sum_{k,k',p} \hat{v}(p)\times \text{(constant terms)}\nonumber\\
&&\quad =  \frac{1}{2N} \sum_{k,k'\in \mathcal{B}_{\mu}} \hat{v}(0) - \frac{1}{2N} \sum_{k,k' \in \mathcal{B}_{\mu}} \hat{v}(k-k')\nonumber\\
&&\quad = \frac{N}{2} \hat{v}(0) - \frac{1}{2N} \sum_{k,k' \in \mathcal{B}_{\mu}} \hat{v}(k-k')\;.
\end{eqnarray}
These last two terms reproduce the Hartree-Fock interaction. This concludes the check of Eq. (\ref{eq:interaction}).

\end{document}